\documentclass[prx, aps, superscriptaddress,groupedaddress,nofootinbib, twocolumn ]{revtex4-2}
\usepackage[utf8]{inputenc}

\usepackage[table,xcdraw,dvipsnames]{xcolor}
\usepackage{amsmath,amsfonts,amsthm,bbm,graphicx,enumerate,hyperref,tikz,pgfplots,ulem}
\usepackage{amssymb}
\usepackage{braket}
\usepackage{physics}
\definecolor{ibm0}{HTML}{648FFF}
\definecolor{ibm1}{HTML}{785EF0}
\definecolor{ibm2}{HTML}{DC267F}
\definecolor{ibm3}{HTML}{FE6100}
\definecolor{ibm4}{HTML}{FFB000}
\usepackage{dsfont}
\usepackage{scalerel}
\usepackage{microtype}

\usepackage{tikz}
\usetikzlibrary{quantikz2}

\usepackage{pgfplots}
\DeclareUnicodeCharacter{2212}{−}
\usepgfplotslibrary{groupplots,dateplot}
\usetikzlibrary{patterns,shapes.arrows}
\pgfplotsset{compat=newest}
\usepgfplotslibrary{groupplots}
\usepgfplotslibrary{dateplot}

\usepackage{stmaryrd}
\usepackage{subcaption}

\usepackage{pifont}
\newcommand{\cmark}{\ding{51}}
\newcommand{\xmark}{\ding{55}}
\usepackage{booktabs}

\usepackage[noend]{algorithm2e}

\def\E{ {\mathcal E} }
\def\P{ {\mathcal{P}} }

\def\R{ {\mathcal R} }
\def\G{ {\mathcal G} }
\def\N{ {\mathcal N} }
\def\F{ {\mathcal F} }

\def\O{ {\mathcal O} }
\def\D{ {\mathcal D} }
\def\S{ {\mathbb{S}} }
\def\T{ {\mathbb{T}} }

\def\I{ {\mathcal I} }

\def\gL{ {\mathbb{L}} }
\def\gP{ {\mathbb{P}} }
\def\gS{ {\mathbb{S}} }

\def\s{\boldsymbol{s}}
\def\m{\boldsymbol{m}}

\def\e{\boldsymbol{e}}
\def\d{\mathbf{D}}

\def\Tr{{\rm{tr}}}
\def\tr{ {\rm{tr}} }

\def\>{\rangle}
\def\<{\langle}

\def\hc{^{\dagger}}

\renewcommand{\emph}{\textit}

\newcommand{\iden}{\mathbb{I}}

\newtheorem{theorem}{Theorem}

\newtheorem{lemma}[theorem]{Lemma}

\newcommand{\nkd}[3]{\left \llbracket #1,\, #2,\, #3 \right \rrbracket}

\hypersetup{
	colorlinks=true,  
	linkcolor=ibm0,   
	citecolor=ibm0,   
	urlcolor=ibm0     
}

\begin{document}
	\title{Characterization of syndrome-dependent logical noise in detector regions}
	\author{Matthew Girling, Ben Criger, Cristina C\^irstoiu} 
	\affiliation{Quantinuum, 13-15 Hills Road, CB2 1NL, Cambridge, United Kingdom}
	\date{August 12, 2025}
\begin{abstract}
Characterizing how quantum error correction circuits behave under realistic hardware noise is essential for testing the premises that enable scalable fault tolerance.
Logical error rates conditioned on syndrome outcomes are needed to enable noise-aware decoding and validate threshold-relevant assumptions.
We introduce a protocol to directly estimate the logical Pauli channels (and pure errors) associated with detector regions formed of two or more syndrome extraction gadgets, conditioned on observing a particular parity in the syndrome outcomes.
The method is SPAM-robust and most suitable for flag-based syndrome measurement schemes.
For classical processing of the experimental data we implement a Bayesian modelling approach.
We validate this new protocol on a small error-detecting code using Quantinuum H1-1, a trapped-ion device, and demonstrate that several noise diagnostic tests for fault tolerance improve significantly when using noise tailoring and mitigation strategies, such as swapped measurements for leakage protection, and Pauli frame randomization. 
\end{abstract}
\maketitle
\section{Introduction}
Fault-tolerant quantum computations are essential to implement large-scale quantum algorithms that enable up-to-exponential quantum speed-ups in a range of applications from simulations of materials to molecular systems \cite{preskill2025beyond}. 
These computations require logical quantum information to be redundantly encoded into physical qubits, with periodic non-destructive measurements resulting in classical outcomes -- \emph{syndromes} --  that indicate the presence of errors and are used by a decoder to infer corrections.
Research into quantum error correction (QEC) \cite{shor1995scheme,gottesman1997stabilizer,knill1997theory,steane1996error,aharonov1997fault, raussendorf2007fault,knill1998resilient, gottesman2016surviving} provides asymptotic guarantees that noise can be exponentially suppressed with poly-logarithmic overhead in device size and execution time.

Hardware advances across different modalities of quantum computers have led to experimental demonstrations of several key QEC components  \cite{chiaverini2004realization,acharya2024quantum,ni2023beating,ofek2016extending,ryan2021realization,postler2022demonstration}. 
Device-specific noise characteristics can strongly alter the effectiveness of a given error-correcting code, the optimal decoding strategy and practical fault tolerance thresholds.
Testing whether current quantum devices satisfy the assumptions required for scalable fault tolerance is essential, as violations can undermine reliability of larger quantum error-corrected computations. 

In order to assess whether a given QEC circuit performs as expected, a characterization protocol should possess two key properties. 
Firstly, such a protocol should be \emph{direct}, inferring logical-level error models using experimental data derived from QEC circuits while making as few additional assumptions (about, e.g. physical-level errors) as possible.
Secondly, QEC component characterization should be \emph{syndrome-dependent}, producing logical error models that are conditioned on the measurement outcomes obtained at runtime. 
These two properties are valuable, as they are required to perform post-selection based on syndrome data \cite{bluvstein2024logical} and soft decoding within a concatenated code \cite{yoshida2024concatenate, knill1996concatenated, meister2024efficient, goto2024many}.
Perhaps most importantly, they allow us to evaluate the validity of premises used within QEC (such as the weight distribution of Pauli errors) using real-world data.
We call characterization methods that satisfy both of these desiderata logical syndrome dependent (LSD) protocols.

There are a variety of existing approaches to characterizing logical errors in QEC gadgets. 
One simple and syndrome-dependent (though indirect) approach to estimating logical error rates is to extrapolate from physical-level characterization using classical simulations \cite{ryan2021realization}.
These methods implicitly assume that errors that occur outside individual gates (such as memory error and crosstalk) are negligibly small, potentially resulting in underestimates of logical error rates.
Furthermore, in order for the fitting process to be tractable, it is typically assumed that leakage errors (where the state of a qubit leaves the computational sub-space) are consistently detected or mitigated \cite{ceasura2022non,flammia2021averaged}, that correlations are limited in strength \cite{wagner2023learning, hockings2025scalable}, and that noise can be faithfully converted into a stochastic Pauli channel form using randomization sub-routines \cite{wallman2016noise, ware2021experimental, beale2023randomized, beale2023randomized2}.
One such protocol, averaged circuit eigenvalue sampling (ACES) \cite{flammia2021averaged}, can simultaneously estimate Pauli channels associated with different operations in a circuit, and has been applied to gate layers as used in syndrome extraction gadgets \cite{fazio2025characterizing, hockings2025scalable, flammia2021averaged, flammia2020efficient, harper2021fast}.
This protocol is scalable, performing well for high distance and qubit numbers, but assumes circuit-level noise models that do not include noise from mid-circuit measurements and resets within the gadget.

An alternative approach to characterizing QEC gadgets is logical randomized benchmarking \cite{combes2017logical}. 
This method is direct (making fewer assumptions about the noise at the physical level) but, after applying a possible correction, the syndrome outcomes are discarded and are not individually incorporated in the estimated average rates.
In addition, logical randomized benchmarking protocols require the implementation of logical Clifford operations, and typically assume that errors in these operations will be gate-independent (though gate-dependent noise has already been addressed in physical RB \cite{wallman2018randomized, helsen2022general}). 

In this work, we directly characterize logical noise and optimal corrections associated with QEC gadgets in detail, conditioned on the syndromes obtained.
Using a unified framework of detectors, we give a SPAM-robust tomographic protocol that simultaneously characterizes sets of Pauli channels associated with detector outcomes given by parities of consecutive syndrome extraction gadgets.
We refer to this particular LSD protocol as LSD-DRT (detector region tomography).
LSD-DRT largely fits into the general framework of randomized benchmarking and (similarly to ACES and cycle benchmarking \cite{erhard2019characterizing}) it requires enforcing a stochastic Pauli error model.
In our case, we analyse extended regions that incorporate two or more stabilizer measurements (including mid-circuit measurement and reset) as well as partial classical information from those measurements that is then used in the estimation procedure.
Our analysis applies also to adversarial noise that goes beyond a typical circuit-level model and may include space-time correlations \emph{within} a gadget (but is otherwise i.i.d. between different gadgets).
In addition, we do not necessarily require local stabilizers.
While the number of measurement settings we require scales exponentially with the  number of logical qubits (similar to other tomographic protocols), we are able to determine not only the distribution of logical errors introduced by a detector region, but also those associated with a pure error (i.e. equivalent up to a stabilizer).  
We use variable-length sequences of syndrome extraction gadgets, from which we construct detectors out of separate pairs of gadgets, and present both frequentist and Bayesian approaches for classical post-processing.
We demonstrate the results for the $[[2,1,1]]$ $X$-error-detecting code using the Quantinuum H1-1 device, and for a $[[4,2,2]]$ code using an emulator.

This experimental demonstration requires that we use several noise-tailoring strategies to ensure the assumption of stochastic Pauli errors is met on average.
We use Pauli frame randomization to twirl coherent errors, and regularly swap ancillas with data qubits in order to limit the accumulation of leakage errors. 
To determine the efficacy of these techniques, we use three diagnostic tests that may be of independent interest, see Section~\ref{sec:noisediagnostics}.

Applying the noise mitigation techniques discussed above to the $[[2,1,1]]$ $X$-error-detecting code, using LSD-DRT, we show a SPAM-robust estimate of logical error rate per two syndrome extraction gadgets with trivial detector outcome of $p(X) = 2(\pm 8) \cdot 10^{-4}  $, $p(Y) = 1(\pm 9) \cdot 10^{-4}  $, $p(Z) = 37(\pm 8) \cdot 10^{-4} $, using post-selection and of $p(X) = 10(\pm 12) \cdot 10^{-4}  $, $p(Y) = 8(\pm 12) \cdot 10^{-4}  $, $p(Z) = 39(\pm 12) \cdot 10^{-4} $ using an ideal decoder.
This shows that when the $[[2,1,1]]$ code is used as a low-level code in a concatenated scheme, it can not only provide a long-lived QEC memory \cite{dasu2025order}, but with suitable noise mitigation, it can do so with improved logical fidelities.

\section{Syndrome extraction gadgets and their conditional logical error probabilities}\label{sec:conditional-qec-channels}

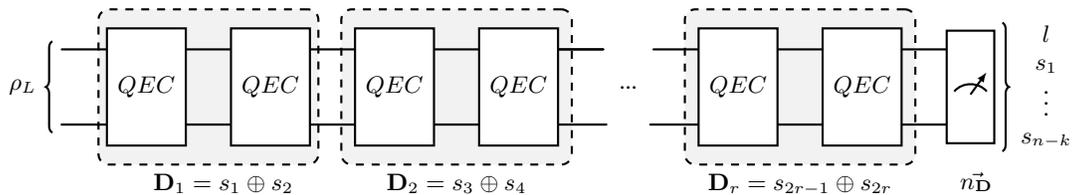
\begin{figure*}[t]
        \centering
		\begin{quantikz}[row sep={1cm,between origins}, column sep=0.6cm]
				\lstick[2]{$\rho_{L}$}  & \gate[wires=2]{QEC} \gategroup[2,steps=2, style={dashed,rounded
					corners,fill=gray!10, inner 
					xsep=0pt},background,label style={label
					position=below,anchor=north, yshift =-0.2cm, xshift=0.2cm}]{{\sc $\d_1 = s_1 \oplus s_2$}}  & \gate[wires=2]{QEC}   & \gate[wires=2]{QEC} \gategroup[2,steps=2, style={dashed,rounded
					corners,fill=gray!10, inner 
					xsep=2pt},background,label style={label
					position=below,anchor=north, yshift =-0.2cm}]{{\sc $\d_{2} = s_3 \oplus s_4$}}  & \gate[wires=2]{QEC}  &\qw \midstick[2,brackets=none]{...} &  \gate[wires=2]{QEC} \gategroup[2,steps=2, style={dashed,rounded
					corners,fill=gray!10, inner 
					xsep=2pt},background,label style={label
					position=below,anchor=north, yshift =-0.2cm}]{{\sc $\d_{r} = s_{2r-1} \oplus s_{2r}$}}    & \gate[wires=2]{QEC}  &  \meter[2]{}  \gategroup[wires=2,steps
		            =1,style={inner sep=6pt, draw=none}, label style = {label position=below, anchor=north,yshift=-0.1cm}]{ \ \ $\vec{n_{\d}}$}    \rstick[2]{ \hspace{0.6cm} $l$ \\ \hspace{0.6cm} $s_1$ \\ \hspace{0.6cm} $\vdots$ \\ \hspace{0.6cm} $s_{n-k}$} \\
		            & \qw  & \qw & \qw & \qw &\qw &\qw  &  \qw & \qw 
		\end{quantikz}
		\caption{ \textbf{Grouping detectors from independent syndrome pairs in sequences of repeated QEC.}
		Repeatedly applying a syndrome extraction gadget $2r$ times results in classical outcomes $\s^1,\ldots,\s^{2r}$.
		From these, by taking the XOR of consecutive syndromes we can obtain the detectors.
		However, in this case, we construct the detectors $\d_1,\ldots, \d_r$ by pairing syndromes disjointly such that any syndrome outcome is not used in more than a single detector.  
		This ensures independence between detectors for circuit level Pauli errors.
		We give a protocol (see Section~\ref{sec:protocol}) to perform a SPAM robust characterization of (Pauli twirled) error channels associated with each outcome of these independent detectors.
		For each shot, we prepare a $+1$ eigenstate $\rho_L$ of a logical Pauli operator $L$, apply a sequence of independent detectors -- recording the total number of each detector outcome  $\vec{n_{\d}} = (n_{\d=\bf{0}}, ...)$ -- before destructively measuring $L$ to obtain $l \in \{0,1\}$ and a complete set of stabilizer generators ($s_1, ..., s_{n-k}$) which allow us to simultaneously measure any logical operator of the form $Q=LS$ for $S \in \gS$, the stabilizer group.
	}
	\label{fig:independent-detectors}
\end{figure*}

An $\nkd{n}{k}{d}$ stabilizer code is described by a commuting sub-group $\S\subset \P_n$ of the Pauli group on $n$ qubits that is generated as $\S= \<S_1,\ldots, S_{n-k}\>$, encoding $k$ logical qubits into $n$ physical qubits. 
Since all elements in the stabilizer group $\S$ commute, they can be measured simultaneously and the joint $+1$ eigenstates determine the codespace of the encoded $k$ logical qubits.
A choice of basis for the codespace may also be specified by a unitary encoding matrix $U_{\textrm{enc}}$ applied to computational basis states $|z_1\ldots z_k\>\otimes |0\>^{n-k}$.
Then, under the encoder's adjoint action, the logical Pauli operators $\gL$ are the image of the Pauli group on the first $k$ qubits (i.e. $U_{\textrm{enc}}( P_{1\ldots k} \otimes I \otimes \ldots\otimes I )U_{\textrm{enc}}^\dagger$).
Equivalently, each element of $\gL$ is a representative of a coset of the stabilizer group within the normalizer group (the set of Pauli operators that commute with all stabilizers).

A projective measurement of a stabilizer $S \in \S$ is given by $\Pi_{s} = \frac{\iden + (-1)^s S}{2}$, with two possible outcomes: $s=0$ projects into the $+1$ eigenspace of $S$ and $s=1$ into the $-1$ eigenspace of $S$.
More generally, projective measurements of all generators of $\S$ will be given by $\Pi_{\bf{s}} =  \prod_{i=1}^{n-k}\frac{ (\iden + (-1)^{s_i} S_i)}{2}$ where ${\bf{s}} = (s_1, \ldots, s_{n-k})$.

The syndrome of an $n$-qubit Pauli error $P$ is given by its commutation relations with each of the generators, namely by the bitstring ${\bf{e}}_P := (\omega(P,S_1), \ldots , \omega(P,S_{n-k}))$.
Here, we use the notation  $\omega(A,B)$ for the bi-character (symplectic product) of two Pauli operators $A, B$, given by $\omega(A,B) = 0 $ iff $[A,B] =0$ and $\omega(A,B) = 1$  iff $ \left \lbrace A,\, B \right \rbrace = 0$.
In particular, given that an error $P$ occurred on a state within the codespace, then the outcome of the stabilizer measurements will deterministically be $\mathbf{e}_{P}$.
Formally, we have the more general useful relation
\begin{align}
P \Pi_{\bf{s}} = \Pi_{\bf{s\oplus e_P}} P.
\end{align}
There are $2^{n-k}$ different syndromes and each will correspond to a coset of $2^{n+k}$ Pauli errors that share that particular syndrome, because any product of a logical operator and stabilizer has trivial syndrome.
We choose a representative for each coset, which we will refer to as a pure error \cite{poulin2006optimal}.
We denote the set of pure errors by $\T$.

For a fixed choice of $\T$ and $\gL$ then any $n$-qubit Pauli operator can be uniquely written as
\begin{align}
P = E \cdot L \cdot S 
\end{align}
where $S \in \S$, $E \in \T$ is the pure error with syndrome ${\bf{e}}_E  = {\bf{e}}_P$ and $L\in \gL$ is a logical Pauli operator.  

Quantum error correction involves repeated, non-destructive measurements of stabilizer group elements.
This gives information on which error has occurred and the Pauli correction that returns the computation back into the codespace.
The role of decoding is to maximize the probability that no logical error occurs when the correction assigned to a particular observed syndrome outcome $\bf{e}$ is applied.
Therefore, we have freedom to apply the correction $E\hat{L}$ where $E$ is the pure error associated to syndrome $\bf{e}$ and $\hat{L}$ is a logical operator given by
\begin{align}
\hat{L} = \arg \max_{L\in\gL} \textrm{prob}( L | {\bf e} ) =  \arg \max_{L\in\gL} \sum_{S\in \S} \frac{\textrm{prob}( E \cdot L\cdot S )}{\textrm{prob}(\bf{e})}.
\end{align}

There are different ways to implement syndrome extraction gadgets (i.e. non-destructive stabilizer measurements).
Flag-based schemes use a small ancilla register to measure a stabilizer with a fault-tolerant circuit that prevents adverse error propagation.
Nevertheless, when a non-trivial error syndrome is observed, one may wrongly infer that a data error occurred and apply an erroneous logical correction, rather than correctly attributing the syndrome to a measurement error.
Using multiple rounds of syndrome extraction allows for decoding algorithms to attribute certain syndromes to measurement errors, allowing these errors to be suppressed similarly to errors on data qubits.
As a result, decoding of a flag-based scheme is based not on a particular observed syndrome outcome $\bf{e}$ but rather on the parities of consecutive measurement outcomes, which we refer to as \emph{detectors}, and denote as $\d$.
Other schemes rely on logical ancillas to simultaneously extract half of the syndrome bits using Steane-style gadgets for CSS codes \cite{steane1996error} or Knill's scheme \cite{knill1996concatenated} that obtains all syndrome bits using a logical Bell pair.
In these cases, errors affecting operations within the QEC gadget have effects equivalent to independent errors on data qubits.
Therefore, individual syndrome outcomes $\bf{e}$ can be used in decoding; effectively setting $\bf{D} = \bf{e}$.

\subsection{Noisy syndrome extraction gadgets as instruments}

The fact that the process of measuring stabilizers is itself subject to errors means that, instead of the target projective measurement $\Pi_{\s}$, a (sub-normalized) pointer map $\I_{\s}$ is implemented whenever the syndrome $\s$ is observed.
A noisy syndrome extraction gadget is then described by an instrument $\I = \{\I_{\s}\} $ that transforms an input state $\rho$ into a classical outcome $\s$ and an updated state $\frac{\I_{\s} (\rho)}{\tr(\I_{\s}(\rho)}$.
For now, we consider circuits that are subject only to stochastic Pauli errors (which may also include adversarial noise models).
If the incoming state $ \rho =\rho_L$ is in the codespace and we observe a syndrome $\s$ with probability $p(\s): =p(\s|\rho_L)$ then this is a consequence of \emph{internal} errors occurring during the syndrome measurement process whose effect is to apply a Pauli error channel $\frac{1}{p(\s)}\P^{\s} (\rho_L)$.
Similarly, if the incoming state already had an error $\rho = M \rho_L M$ with syndrome $\m$ and we observe syndrome $\s$, this means that the internal errors during the syndrome measurement correspond to  $\s \oplus \m$, the difference between the observed syndrome and the syndrome of the incoming error.
In this case, the effect of the noisy syndrome measurement is given by the Pauli channel $\frac{\P^{\s\oplus \m}}{p(\s\oplus\m)} (\rho)$.
Combined, the effect of noise in the QEC gadget can be described by a collection of (sub-normalized) Pauli channels  $\{\P^{\s}\}$ with different Pauli error probability distributions $\{p^{\s}(P)\}_P$ associated with each possible syndrome outcome with normalization $\sum_{P} p^{\s}(P) = p(\s)$.
A noisy QEC gadget can then be viewed as a process that modifies an incoming probability distribution of errors depending on the observed syndrome, 
\begin{tikzpicture}
	\node at (-0.1,1) {$ \left\lbrace  p_{in} (M)\right\rbrace _{M} $};
	\draw[-{Stealth[length=2mm, width=2mm]}] (0.9,1) -- (1.3,1);
	\draw[thick] (1.4,0.5) rectangle (2.6,1.5);
	\draw[] (1.9,1.5) -- (1.9,1.8);
	\draw[] (2.1,1.5) -- (2.1,1.8);
	\node[align=left] at (2,1) {$QEC$};
	\node at (2,2) {$\s$};
	\draw[-{Stealth[length=2mm, width=2mm]}] (2.7,1) -- (3.1,1);
	\node at (5.2, 1) {$ \left\lbrace \sum\limits_{P} p_{in} (P) p^{\m \oplus\s} (P M)\right\rbrace _{M}  $};
\end{tikzpicture}.

Formally, since an arbitrary incoming state $\rho$ will have a mixture of different errors the noisy syndrome measurement with outcome $\s$ is most generally described by
\begin{equation}
	\I_{\s}(\rho) = \sum_{\m} \P^{\m\oplus \s} (\Pi_{\m} \rho \Pi_{\m})
\end{equation}
with the outgoing state conditioned on the observed syndrome given by $\frac{\I_{\s} (\rho )}{\Tr(\I_{\s} (\rho))} = \frac{\sum_{\m} \P^{\m\oplus \s} (\Pi_{\m} \rho \Pi_{\m})}{\sum_{\m} p(\m\oplus \s) \Tr(\rho \Pi_{\m})}$. 
If we were to then apply a correction $\R(\s)$ based on trusting the observed syndrome, then the QEC gadget and decoder would be described jointly as $\R \I = \{ \R(\s) \circ \I_{\s} \}_{\s}$. 

\subsection{Detector regions for repeated noisy QEC gadgets}
As discussed above, we will consider repeated rounds of two (or more) error correction gadgets with a fixed detector outcome $\d = \s_1\oplus \s_2$. 
If the incoming state $\rho_L$ is in the logical codespace and we observe syndromes $\s_1$ and $\s_2$ with probability $p(\s_1,\s_2)$, then the errors occurring internally within the detector region lead to the output state
\begin{equation}\label{eqn:syd-error-channel}
	\begin{split}
		\E^{(\s_1, \d\oplus \s_1)} (\rho_L) :=& \frac{\I_{\s_2} \I_{\s_1}(\rho_L)}{p(\s_1,\s_2)} \\
		=& \frac{ \left(\sum_{\m_2} \P^{\m_2\oplus \s_2} \P^{\s_1}_{\m_2} \right)(\rho_L)}{p(\s_1,\s_2)},
	\end{split}
\end{equation}
where $\P^{\s_1}_{\m_2}$ denotes the terms of the Pauli channel $\P^{\s_1}$ that contain only the errors with syndrome $\m_2$. 
Now, if the incoming state $\rho = M\rho_L M$  had an error with syndrome $\m$ but we were to observe $\s_1$ and $\s_2$, then the output state will be affected by a different channel arising from the internal errors within the detector region compatible with the classical outcomes, leading to the output state $\E^{(\s_1\oplus \m, \d \oplus\s_1 \oplus \m)}(\rho)$.

The \emph{detector error channel} describes only the effect of internal errors that occurred within the detector region conditioned on observing a classical detector outcome $\d$, and does not depend on the errors of the incoming state. 
It can be calculated by summing over events that lead to the same detector outcome $\d$ 
\begin{equation}
	\E^{\d} =  \frac{1}{p(\d)} \sum_{\s_1} p(\s_1, \d\oplus\s_1)\E^{(\s_1,\d\oplus \s_1)}.
\end{equation}
In terms of the Pauli channels $\{\P^{\s}\}$ that describe a single noisy QEC circuit, the detector error channel is equivalently given by
$ \E^{\d} = \frac{1}{p(\d)} \sum_{\m_1,\m_2} \P^{\m_2}\P^{\m_1}_{\m_1 \oplus \m_2 \oplus \d}$.
In Section~\ref{sec:protocol} we introduce a SPAM-robust protocol to determine, up to stabilizers, the probability of errors in $\E^{\d}$ that are introduced within the detector region, conditioned on observing a given detector outcome $\d$. 

\subsubsection{Example: Phenomenological model}
To illustrate the description of noisy QEC gadgets as instruments and detector error channels we consider a small $[[2,1,1]]$ $X$-error-detecting code  with stabilizer $-ZZ$ that gives two possible syndrome outcomes $\s \in \{0,1\}$ corresponding to two-qubit Pauli errors that respectively commute or anti-commute with $-ZZ$. 
The pure error is given by $IX$ and logical errors by$\bar{X} = XX$ and $\bar{Z} = ZI$.  
We consider a phenomenological noise model where independent $X$ errors occur with probability $p$ on the data qubits and syndrome bits are incorrect with probability $q$.
The channel on the data qubits takes the form
\begin{equation}
	\begin{split}
		\Lambda(\rho) &= \underbrace{(1-p)^2\rho + p^2 XX\rho XX}_{\Lambda_{I}} \\ & \hspace{1cm} + \underbrace{(1-p) p \ IX (\rho  + XX\rho XX) IX}_{\Lambda_{IX}}. 
	\end{split}
\end{equation}
Then the noisy measurement process when a syndrome $\s$ is observed will be given by
\begin{equation}
	\begin{split}
		\I_{\s} (\rho)  &= [\underbrace{(1-q)\Lambda_{I} + q \Lambda_{IX} }_{\P^0}](\Pi_{\s} \rho \Pi_{\s})  \\ & \hspace{1cm} + [\underbrace{q \Lambda_{I} + (1-q) \Lambda_{IX}}_{\P^1}] (\Pi_{\s\oplus 1}\rho \Pi_{\s\oplus 1}).
	\end{split}
\end{equation}
Note that $\P^0$ and $\P^1$ are sub-normalized maps, but $\P^0 + \P^1 = \Lambda$ is a proper trace-preserving channel which can be estimated by discarding the classical syndrome outcomes.

To construct the detector error channels, from eqn.~(\ref{eqn:syd-error-channel}) we have
\begin{align}
	p(0,0)\E^{(0,0)}  &=  (1-q)^2 \Lambda_I^2  + q \Lambda_{IX} \Lambda_I+  (1-q)q \Lambda_{IX}^2 \\
	p(1,1)\E^{(1,1)}  &=  q^2 \Lambda_I^2  + (1-q) \Lambda_{IX} \Lambda_I+  (1-q)q \Lambda_{IX}^2
\end{align}
and therefore the detector error channel conditioned on $\d =0$  -- which arises from the two events above where either $(\s_1, \s_2)$ is  $(0,0)$ or $(1,1)$  -- will be given by
\begin{align}
	\E^{\d = 0} = \frac{[(1-2q + 2q^2)  \Lambda_I^2 + 2(1-q) q \Lambda_{IX}^2]  +  \Lambda_{IX} \Lambda_{I}}{p(\d = 0)} .
\end{align}
In the next section, we show -- in more generality -- how detector error channels can be directly estimated with a protocol that incorporates classical syndrome outcomes. 

\section{SPAM-robust learning of syndrome-dependent logical Pauli channels}\label{sec:protocol}

Randomized benchmarking (RB) protocols typically amplify errors by sequential applications of target operations, using different sequence lengths so as to separate the SPAM errors and produce robust estimates.
Here, we describe such a robust protocol to determine detector error channels associated with regions of two (or more) syndrome extraction gadgets and conditioned on specific detector outcomes.
The main assumption in the derivation is that these channels are given by stochastic Pauli errors.
To ensure this condition is met in practical implementations, we modify the standard stabilizer measurement circuit as detailed in Section~\ref{sec:noisediagnostics}.

Twirling over the Pauli group (one of the modifications used) transforms a general quantum error channel $\D^{\d}$ associated with detector $\d$ into a Pauli channel 
$\E^{\d} (\rho) : =\mathbb{E}_{P\in \gP_n} P \, \D^{\d} (P \rho P)  P$, that takes the form
\begin{align}
\E^{\d} (\rho) = \sum_{P\in \gP_n}  p(P|\d) P\rho P.
\end{align}
Equivalently, a Pauli channel can be fully described by its eigenvalues $\lambda_{\d}(P)$ that satisfy $\E^{\d}(P_a) = \lambda_{\d}(P_a) P_a$.
The set of all eigenvalues $\{\lambda_{\d}(P)\}_{P\in \gP_n}$ is equivalently related via the $\mathbb{Z}_2$ Fourier (or Walsh-Hadamard) transform to the set of error probabilities of Pauli operators $\{p(P |\d)\}_{P\in \gP_n}$. 

Pauli errors that are equivalent up to multiplication by an element of the stabilizer group are identical at the logical level.
Therefore, we are interested in the coarse-grained probability distribution over stabilizer cosets, from which we can determine both the logical error channel and correction for a given set of detector outcomes.
Specifically, the probability of a coset $P\S$ will be given by
\begin{align}
p(P\S | \d)  := \frac{1}{|\S|} \sum_{S\in \S} p(PS |\d)
\end{align}

Equivalently, these probability distributions relate to eigenvalues over logical errors via the Fourier transform
\begin{align}
p(P\S | {\bf{D}}) = \frac{1}{2^{n+k}} \sum_{Q \in \gL\S} \lambda_{\d}(Q) (-1)^{\omega(P, Q)},
\end{align}
where we note that we only require eigenvalues corresponding to those Pauli operators $Q$ that commute with all elements in the stabilizer group (i.e. the normalizer of $\S$), as a consequence of Lemma~\ref{lem:eigenvalueslogical}.
Therefore, learning the set of eigenvalues $\lambda_{\d} (Q)$ for all $Q \in \mathbb{L} \mathbb{S}$ suffices to determine probabilities of error cosets for the given detector $\d$.
Furthermore, the eigenvalues relate to expected values of measuring operators $Q$ as 
\begin{align}
	\mathbb{E} (Q| \d,\rho ) = \Tr(Q \E^{\d} (\rho))  = \lambda_{\d}(Q) \Tr(Q\rho),
\end{align}
for some input state $\rho$.
When $\rho$ is a +1 eigenstate of $Q$ then the expected value of $Q$ given detector outcome $\d$ is simply the eigenvalue $\lambda_{\d}(Q)$. 

\subsection{Main characterization algorithm}
The LSD-DRT algorithm is an RB-like protocol that incorporates syndrome information in repeated rounds of QEC gadgets at various lengths to robustly determine the detector error channels $\E^{\d}$, up to stabilizers.
It involves $2r$ repeated applications of syndrome extraction with outcomes $\s^{(i)} = (\s_1,\ldots,\s_{2r})$, from which we define \emph{independent detectors} $\d_1,\ldots, \d_r$ between disjoint pairs of syndromes $(\s_1, \s_2), (\s_3,\s_4), \ldots (\s_{2r-1}, \s_{2r})$.
This involves partially marginalizing over the classical outcomes, but it has the effect that the post-measurement state after the $2r$ rounds where we observe the detectors $\d_1 = \s_1\oplus \s_2$, \ldots, $\d_r = \s_{2r-1} \oplus \s_{2r}$ is given by $\E^{\d_r} \circ \ldots \circ \E^{\d_1} (\rho)$ for any input state $\rho$. Therefore, the expectation value of measuring a Pauli observable $Q$ takes the form
\begin{align}\label{eqn:pure-pauli-fit}
	\mathbb{E}(Q|\d_1, \ldots, \d_r) = \lambda_{\d_1}(Q)\lambda_{\d_2}(Q) \ldots \lambda_{\d_r}(Q) \Tr [Q\rho].
\end{align}
We prove these results in Appendix~\ref{ap:mainprotocol}.
The assumptions going into the derivation of this fitting model are weaker than those commonly encountered in QEC.
Specifically, we assume that Pauli errors occur in the syndrome extraction gadget (both in the $n$ data qubits and the measured ancillas) and that their distribution is independent and identical between gadgets.
This includes circuit-level noise models, but also allows for correlated and non-i.i.d. errors within a gadget.

If the state preparation and measurement are affected by Pauli channels, the expected value of an operator $Q$ given that detectors $\d_1,\ldots,\d_r$ were observed and the input target state is a +1 eigenstate of $Q$ will then be given by
\begin{align}\label{eqn:analytical-fit-function}
	\mathbb{E}[Q|\d_1,\ldots,\d_r] = A_{Q}\prod_{\d} (\lambda_{\d} (Q) )^{n_{\d}},
\end{align}
where $A_Q$ incorporates the SPAM errors and the product ranges over all unique detector values $\d \in \{\d_1, \ldots, \d_r\}$ that were observed, where each value occurs $n_\d$ times in the given detector sequence with $\sum_{\d} n_{\d} = r$.
From the above, we also explicitly see that the expected value of $Q$ depends only on the set of counts $\vec{n}_{\d} \equiv \{ n_{\d} \}$ and not on the specific sequence of observed detectors $\d_1,\ldots, \d_r$.
For this reason, we can streamline the notation by using $Q|_{\vec{n_{\d}}}$ to denote the Bernoulli random variable given by the outcome of measuring the Pauli operator $Q$ conditioned on observing the detector counts $\vec{n}_{\d}$, which has expected value $\mathbb{E}(Q|_{\vec{n}_{\d}}) = \mathbb{E}(Q|\d_1,\ldots,\d_r)$.
To account for other sources of non-Pauli errors such as leakage in state preparation and measurement we will allow for a small offset parameter in the fitting model for each observable, which we justify in Appendix~\ref{sec:non-pauli-effects}.

The LSD-DRT estimation protocol (using frequentist statistics) has the following steps:
\begin{enumerate}
	\item \emph{Prepare a $+1$ eigenstate of a logical operator $L\in \gL$ and apply a sequence of $2r$ error correction gadgets, followed by a destructive measurement of $L$. }
	\item \emph{Obtain measurement outcomes of $Q = LS$ for all stabilizers $S\in \S$.}
	Since $L$ commutes with all stabilizers, they can be simultaneously measured and therefore the final measurement can give not only an outcome $l$ from measuring $L$ but also the outcomes of measuring each stabilizer, forming the syndrome ${\bf{o}} \in \{0,1\}^{ \times n-k}$.  For a CSS code, half of the stabilizers can be measured transversally, simultaneously with a final logical measurement \cite{gottesman2016surviving}.
	However, this step does not have to be fault-tolerant, due to the SPAM-robustness of the protocol, so it can be applied more generally without increasing the number of measurement settings.
	For example, the destructive measurement of $L\in \gL$ can be implemented by first applying the decoding circuit $U_{\textrm{enc}}\hc$ followed by an appropriate basis change.
	Then the first $k$ qubit outcomes $|z_1...,z_k\>$ will give $l = \frac{\<z_1,...z_k| L |z_1,...z_k\> + 1}{2} \in \{0,1\}$ and the last $n-k$ qubit outcomes will give ${\bf{o}} \in \{0,1\}^{\times n-k}$.
	Any $S\in \S$ can be written in terms of the generators as $S = S^{a_1}_{1} ... S^{a_{n-k}}_{n-k}$ for a binary string  ${\bf{a}} = (a_1,... ,a_{n-k})$, so it follows that the outcome measurement of $Q=LS$ will be given by $q := (l \oplus {\bf{o}\cdot \bf{a}} \mod 2)$.
	\item \emph{For different sequence lengths $r_1,..., r_M$, repeat the measurement $N_1,...,N_M$ times, gather detector statistics and final outcome data to construct sample mean estimators $\bar{q}|_{\vec{n}_{\d}}$ of the expected value $\mathbb{E}(Q|_{\vec{n}_{\d}})$ for each unique set of detector counts $\vec{n}_{\d}$ observed.}
	\item \emph{Estimate eigenvalues of the detector error channels $\{\lambda_{\d}(Q)\}$ and SPAM parameters by fitting the data to the decay model  $\mathbb{E}(Q|_{\vec{n}_{\d}}) = A_Q \prod\lambda_{\d}^{n_\d} (Q) +B_Q$.} 
	\item \emph{Repeat the above procedure (1-4) for all logical operators $L\in\gL$ to obtain $\{\lambda_{\d}(Q)\}_{Q\in \gL \S}$.}
	\item \emph{Determine the error probabilities $p(P\S|\d)$ by applying the Walsh-Hadamard transform to the estimated eigenvalues.  }
\end{enumerate}
For the Bayesian approach given in  Section~\ref{sec:bayes}, we replace steps (3)-(5) with sampling from the model given in Figure~\ref{fig:bayes-model}.
 
\subsection{Sample complexity}
The complexity of the LSD-DRT algorithm scales exponentially with the number of \textit{logical} qubits.
The number of measurement settings required is $\O(3^k)$, given the measured logical Pauli operators are grouped into commuting sets.
Notably, our protocol relies on expected values of Pauli operators in the larger set, $ \gL \cdot \S$ however, as stabilizers commute with each other and logical operators, all these stabilizer outcomes and final logical measurements can be measured simultaneously. 
In this manner we maximize the information learned from each shot.

\section{Data processing}\label{sec:data-process}

We now turn to how to estimate the eigenvalues and SPAM parameters with classical post-processing.
We give two approaches to this with different advantages, the first using frequentist and second using Bayesian statistics.
This allows us to reduce experimental overhead by building in some prior knowledge of the parameters.
We build on techniques originally formulated for randomized benchmarking \cite{hincks2018bayesian} which we adapt to our setting.

\subsection{Frequentist approach}

From the experiment data,  we obtain the sample mean estimators $\bar{q}|_{\vec{n}_{\d}}$of the expectation value $\mathbb{E}(Q|_{\vec{n}_{\d}})$ for each unique set of detector counts $\vec{n}_{\d}$ that we observe.
For sequences of $r$ detectors the probability of observing a given vector of detector counts follows the multinomial distribution described in Appendix~\ref{sec:sample-complexity}, and therefore the variances $\mathbb{V}[\bar{q}|_{\vec{n}_{\d}}]$ of the estimators will not be uniform for different $\vec{n}_{\d}$.
To deal with the heteroscedasticity in the data, we apply a weighted least-squares fit with weights proportional to the inverse of the variance, which can be obtained either from bootstrapping or analytically.
The decay function will be a modified version of that given in eqn.~(\ref{eqn:analytical-fit-function}), with an off-set to account for leakage errors within measurement.
We seek to solve for each eigenvalue $\lambda_{\d}(Q)$ for $ Q \in \mathbb{L}\mathbb{S}$ using the following minimization:
\begin{equation}\label{eqn:least-squares-fit}
	\min_{\substack{A,B, \{\lambda_{\d}(Q)\}}}   \sum_{ \vec{n}_{\d} }  \! \frac{ \big( 	\bar{q}_{\vec{n}_{\d}} -(A (\prod_{{\bf{D}}}  \left(\lambda_{\d}(Q) \right)^{n_{\d}} + B) \big)^2}{\mathbb{V}(\bar{q}|_{\vec{n_{\d}}})}, 
\end{equation}
where we loosely bound the SPAM parameters as $A + B \leq 1$ and eigenvalues as $-1 \leq \lambda_{\d}(Q) \leq 1$.
The uncertainties in the estimated eigenvalues can be propagated through the final transform straightforwardly.
In the situation where all noise is Pauli, estimation of the simpler fit in eqn.~(\ref{eqn:analytical-fit-function}) admits an analytical solution, and in Appendix~\ref{sec:sample-complexity} we discuss the covariance matrix describing the uncertainty in the estimated eigenvalues. In particular, this can be used to design the experiment practically by selecting sequence lengths that reduce the variances in the estimated probabilities of Pauli errors. 

\begin{figure*}[t]
	\includegraphics[width=15cm]{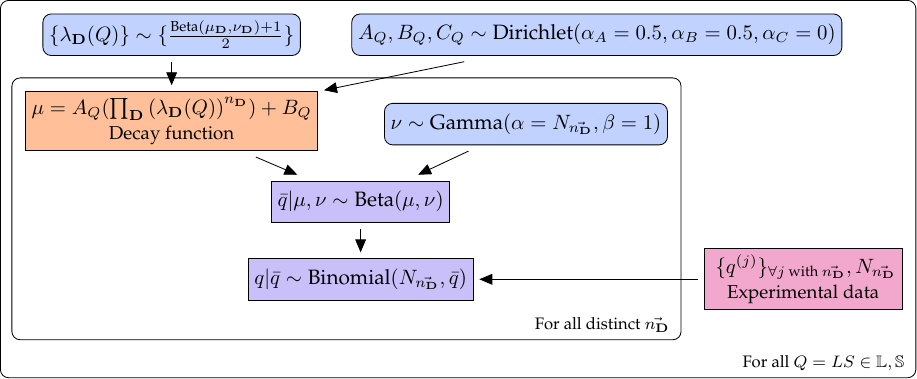} 
	\caption{ \textbf{Hierarchical Bayesian model.}
	Prior probability distributions ($\color{ibm0!50} \blacksquare$) are shown for the set of eigenvalues for each detector outcome channel, $\E^{\d}$, SPAM parameter sets and hyperparameter sets.
	Through the structure of the parameterized model, and the analytical formula of the decay function ($\color{ibm3!50} \blacksquare$), these priors can be updated (using Bayes' rule) given the likelihood distributions ($\color{ibm1!50} \blacksquare$) generated from experimental data ($\color{ibm2!50} \blacksquare$).
	This produces a posterior distribution across all parameters from which we infer marginals for the parameters of interest.
	In practice, we sample from the posterior using Markov chain Monte Carlo methods \cite{abril2023pymc, hincks2018bayesian}.
	}
	\label{fig:bayes-model}
\end{figure*}

\subsection{Bayesian approach}\label{sec:bayes}

A common alternative to frequentist statistics, the Bayesian implementation of data processing takes prior probability distributions for each detector channel as input, and updates them using the experimental data gathered to produce marginal posterior distributions.
Given the form of eqn.~(\ref{eqn:analytical-fit-function}) we can construct a hierarchical model to relate discrete shot outcomes to parameters of interest \cite{martin2024bayesian, blume2010optimal, huszar2012adaptive, hincks2018bayesian, abril2023pymc}.
In Figure~(\ref{fig:bayes-model}) we show the complete parameterized model.

Focusing on a particular logical operator, $Q$, for every unique set of detector clicks $\vec{n}_{\d}$, we collect a set of outcomes $\{ q^{(j)} \}$ where every $q^{(j)} \in \{0,1\}$ and $\abs{\{ q^{(j)} \}} = N_{\vec{n}_{\d} }$, the number of shots for this set of detector clicks.
Given the experimental data we can define a discrete likelihood distribution for every set of unique detector outcomes, $\vec{n}_{\d}$, such that
\begin{equation}
	q|\bar{q} \sim \text{Binomial}( N= N_{\vec{n}_{\d} }, \bar{q}  )
\end{equation}
where the sample size of the distribution is given by the observed shot count $N_{\vec{n}_{\d} }$.
The parameter $\bar{q}$ is a latent variable representing the success probability of any particular Bernoulli trial ($q^{(j)}$) for this binomial distribution.
As any success probability can be described by a bounded distribution on the unit interval, $\bar{q}$ is described by the canonical distribution 
\begin{equation}\label{eqn:beta}
	\bar{q}|\mu, \nu \sim \text{Beta}(\mu, \nu)
\end{equation}
which we parametrize in terms of the mean, $\mu$, and the sample size parameter $\nu > 0$.

We now turn to defining priors for the remaining parameters which are updated given the likelihood distributions, as shown in Figure \ref{fig:bayes-model}.
For every unique $\vec{n}_{\d}$, the sample size parameter $\nu$ is a hyperparameter for which we should choose a non-informative prior with support on all valid values.
A gamma distribution is suitable since it has only on the positive real numbers, and we use as a prior
\begin{equation}
	\nu \sim \text{Gamma}(\alpha = N_{\vec{n}_{\d} }, \beta = 1)
\end{equation}
with shape parameter $\alpha$ as $N_{\vec{n_{\d}} }$ and rate parameter $\beta = 1$.
This ensures the distribution is centred on the true sample size (the observed shot count) for the experimental dataset, but not sharply peaked.

For each $\vec{n}_{\d}$, the distribution $\mu$ can be directly related to the set of eigenvalues of each detector channel $\{ \lambda_{\d}(Q) \}$ and SPAM parameters ($A_Q,B_Q$) by
\begin{equation} \label{eqn:bayes-fit}
	\mu = A_Q \left( \prod_{{\bf{D}}} \left(\lambda_{\d}(Q) \right)^{n_{\d}} \right) + B_Q
\end{equation}
through eqn.~(\ref{eqn:analytical-fit-function}), assuming the assumptions of the protocol hold. 
The SPAM parameters must always obey some simple bounds which ensure $0 \leq \mu \leq 1$.
Further, in the realistic regime we have $A,B>0$ and $A \approx B$. 
Therefore we use a mildly informative prior 
\begin{equation}
	A_Q, B_Q, C_Q \sim \text{Dirichlet}(\alpha_A = \alpha_B=0.5,\alpha_C = 0)
\end{equation}
where $C_Q$ is an extra hyperparameter to allow $A + B \leq 1$.

For the set of eigenvalues $\{ \lambda_{\d}(Q) \}$, to ensure they have support across the whole parameter space, we use priors given by a scaled beta distribution
\begin{equation}
	\{ \lambda_{\d}(Q) \} \sim \left \lbrace  \frac{\text{Beta}(\mu_{\d}, \nu_{\d}) + 1}{2} \right \rbrace.
\end{equation}
where $\mu_{\d}$ and  $\nu_{\d}$ are again mean and sample size parameters respectively.
Settings for $\{\mu_{\d}, \nu_{\d} \}$ for each operator $Q$ complete the model.

To produce a set of uninformative priors for the eigenvalues, we set $\mu_{\d}=1/2$ and $\nu_{\d}=2$ which is equivalent to the uniform distribution.
For an informative prior, we would set $\{ \mu_{\d}, \nu_{\d} \}$ using information from previous experimental runs or statevector simulations of small systems (see Figure~\ref{fig:LSD-2,1,1-bayes}), thereby reducing the experimental overhead required to achieve the same credible interval for the final estimates.
For the experiments in this work, we construct such priors using the emulated device H1-1E.
We find performing LSD-DRT using similar shot counts and sequence lengths for the emulator as for the real device produces posterior distributions which are amenable to use as experimental priors.

In order to obtain our final estimates, we perform sampling using Markov chain Monte Carlo methods using the python package \texttt{pymc} \cite{abril2023pymc}. The results for both post-processing techniques are given in Section~\ref{sec:results}. Possible extensions to the model are discussed in Appendix~\ref{sec:channel-probs-issues}.

\section{Noise mitigation and diagnostics for fault tolerance} 
\label{sec:noisediagnostics}
\begin{figure*}[t!]
	\includegraphics[width=\textwidth]{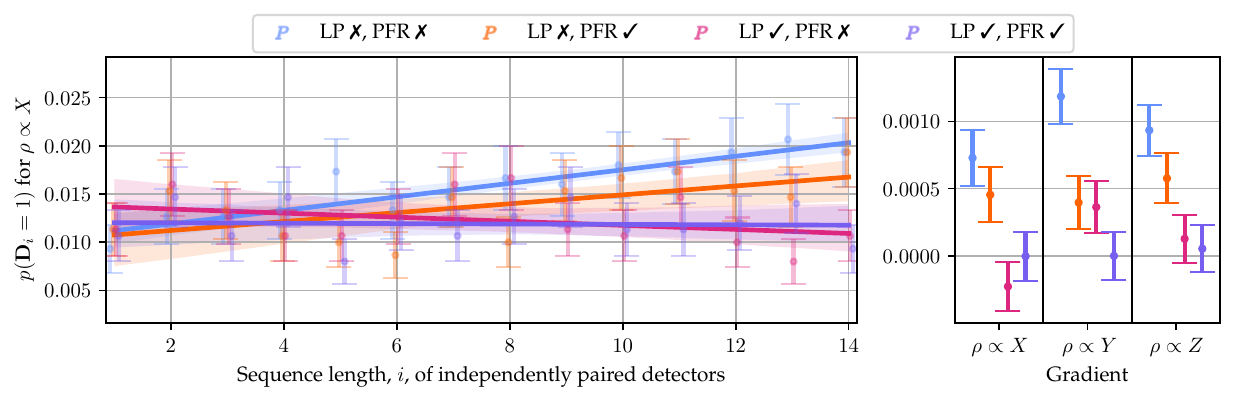}
	\caption{
	\textbf{Stabilized detector click rate for the $\nkd{2}{1}{1}$ code using noise-suppression strategies.}
		Using H1-1, we compare the detector click rate over time with and without both leakage protection (LP) and Pauli frame randomization (PFR).
		A linear least-squares fit is performed for each input state ($\rho \propto X, Y, Z$) and we plot one arbitrary input state ($X$) across settings.
		The estimated gradient for each dataset is given, where a statistically significant increase in click rate over time is observed without LP.
		Additionally, PFR can be seen to ensure a consistent the click rate between different input states.
		Experimental details given in Table \ref{tab:exp-settings}.
	}\label{fig:leak-increase-click-rate}
\end{figure*}

The estimation procedure introduced in Section~\ref{sec:protocol}, requires that the error channels can be efficiently twirled into stochastic Pauli errors.
Errors are commonly assumed to be of stochastic Pauli form within QEC -- both when calculating theoretical thresholds and during decoding \cite{beale2018coherence}.
Not making this assumption can also lead to thresholds an order of magnitude higher \cite{kueng2016comparing}.
In this section, we discuss active noise tailoring and suppression techniques which can be used to achieve statistics compatible with a stochastic Pauli error model.

The following observations allow us to identify non-Pauli channel behaviour by examining the statistics of detector counts within sequences.
Within a sequence of independent detectors (see Figure~\ref{fig:independent-detectors}), the probability of a particular detector outcome -- such as $p(\d = \bf{0})$ -- should not depend on: (1) the initial input state $\rho$ of the sequence, (2) the position of detector within the sequence, or (3) the outcome of any other detector.
Concretely,
\begin{align}
	(1) \, &p(\d_i  | \rho )  = p (\d_i ) \\
	(2) \, &p(\d_i  )  = p(\d_j) \\
 	(3) \, &p(\d_i , \d_j)  = p(\d_i) p(\d_j)
\end{align}
These properties essentially follow from the fact that logical states are $+1$ eigenstates of all elements the stabilizer group and the proof of Lemma~\ref{eq:product-independent-detectors}.
There are a series of noise mechanisms that can lead to violations of the above relationships, including non-Pauli errors (see Appendix~\ref{sec:non-pauli-effects}) and time-dependent errors.

\subsection{Leakage protection via SWAPs}

Quantum computer architectures with qubits defined using two levels of a multi-level system \cite{moses2023race, mcewen2021removing, muniz2025repeated} are susceptible to leakage -- where the state of the physical system in question leaves the computational subspace.  
In Quantinuum devices, gate operations act as the identity on ions outside the computational subspace, and measurements will give fixed outcomes.
Therefore, within a circuit, leakage errors can propagate and accumulate over time, strongly limiting the effectiveness of QEC.

Leakage can be detected with (non-SPAM robust) leakage detection gadgets \cite{stricker2020experimental, moses2023race}.
However, this requires an increased rate of post-selection with circuit size.
Several methods that aim to actively suppress leakage in QEC circuits include direct removal \cite{battistel2021hardware, miao2023overcoming}, introducing auxiliary qubits to apply leakage reduction units \cite{fowler2013coping, brown2019handling} and feedback \cite{bultink2020protecting}, or swapping the role of data and ancilla measurement qubits 
\cite{brown2019leakage}. 
Here, we use the latter approach and periodically measure out all physical qubits using leakage-protected syndrome extraction circuits that have additional SWAPs, as illustrated in Figure~\ref{fig:lp-fault-tolerant-211-gadget}.
In general, using this technique, for an $n$-qubit scheme with $m$ measurements per gadget a leaked qubit is returned to the computational subspace within (at most) $\frac{n+m}{m}$ subsequent gadgets.

We can assess whether leakage protection (LP) via swaps is effective at removing time-dependent errors by measuring correlations between independent detectors, including the use of two-point correlators  \cite{spitz2018adaptive}.
Further, using state-vector simulations, we show that typical leakage rates can cause strong correlations (see Figure~\ref{fig:marginal-prob-detectors}) and a rate of non-trivial detector outcomes that increases with the number of QEC gadgets.

\subsection{Pauli twirling a syndrome extraction gadget} 
Given a leakage-protected circuit implementation of the syndrome measurements, we further use Pauli frame randomization (PFR) \cite{ware2021experimental, hashim2020randomized, wallman2016noise} to twirl the residual errors into Pauli channels, which we can then estimate using LSD-DRT.
For every gadget in a sequence, we apply a uniformly random physical Pauli operator, $P$, on the data qubits before and after.
During circuit compilation, we combine the two adjoining Pauli operators between gadgets into a depth-1 operation.
A classical correction must be applied to the syndrome data $\hat{\bf{e}} = {\bf{e}} \oplus {\bf{e}}_P$ for every observed outcome $\bf{e}$ and where $ {\bf{e}}_P$ is the syndrome of the sampled Pauli operator, introduced before the gadget.

The effectiveness of PFR on gadgets can primarily be assessed tomographically by preparing and measuring in a different Pauli basis, which should be statistically indistinguishable from zero. We can also compare average detector rates across different input states, where significant variations can indicate a coherent component to the error channels (see Appendix~\ref{sec:non-pauli-effects}).

\subsection{Leakage twirling}
Depending on where leakage occurs in the circuit, the use of the two techniques above can cast local leakage errors into a Pauli channel form.
Using LP, any physical qubit will eventually be read out.
When we include PFR, using an example of measuring a $Z$ operator on a single qubit, we essentially have the following circuit fragment
\begin{equation}
    \begin{quantikz}[row sep=0.2cm, column sep = 0.2cm]
      \lstick{A: }  & \gate{P}       & \targ{}    & \ctrl{1}   &  \meter{} \rstick{$\scriptstyle \pm Z$}   \\
       \lstick{B: } & \setwiretype{n}  {\ket{0}} &  \ctrl{-1} \qw  &  \targ{} \qw & \gate{P} \qw & \qw \\
    \end{quantikz}
\end{equation}
up to some single-qubit rotations depending on the QEC code.
If a leakage event has occurred on qubit $A$ before the circuit or during the first CNOT then neither CNOT is applied.
This may induce some error on qubit $B$ but this will not depend on the previous state of qubit $A$ \cite{hayes2020eliminating}.
Qubit $B$ then has a random Pauli operator applied to it.
As $\mathbb{E}_{P \in \gP_1} P \rho P = \mathds{1}/2$ for any single qubit state $\rho$, this depolarizes qubit $B$.

Having described our noise mitigation techniques and how we can witness non-Pauli channel behaviour analytically, we turn to our experimental implementation. 

\section{Implementation and  Numerical results}\label{sec:results}

\begin{figure*}[t]
	\begin{subfigure}{0.99\textwidth}
		\centering
		\includegraphics[width=\textwidth]{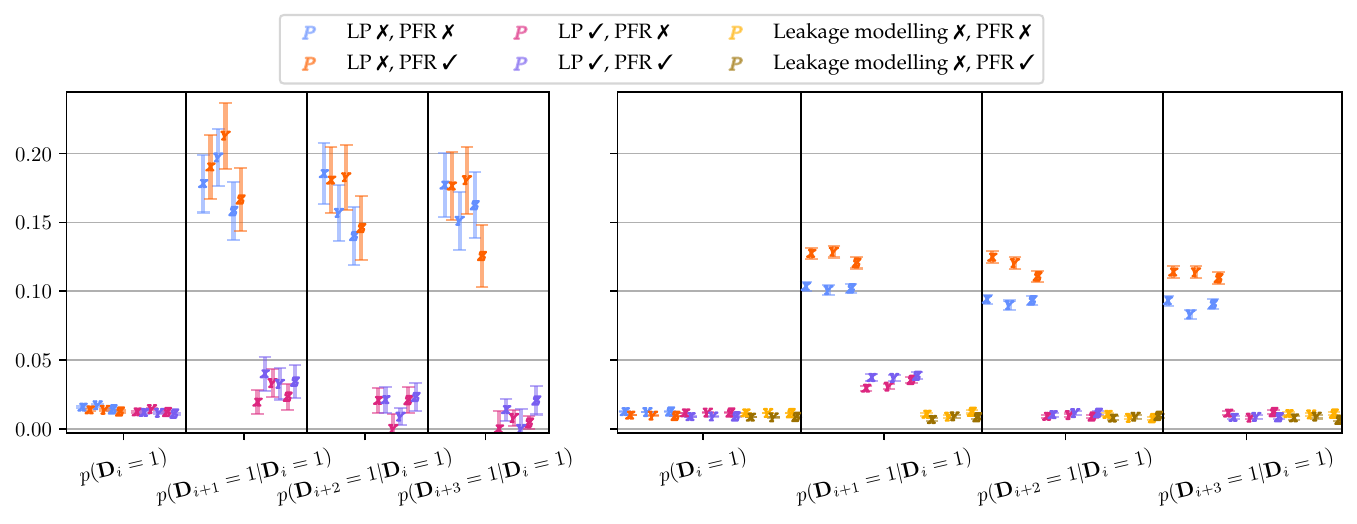}
		\caption{ \textbf{ H1-1} \hspace{7.2cm} (b) \textbf{ H1-1E} \hspace{1cm} }
	\end{subfigure}
	\caption{ \textbf{Measurement of correlations in detector outcomes for $\nkd{2}{1}{1}$ code.}
		(a) For H1-1, conditional probabilities of non-trivial detector outcomes are estimated for a sequence of 15 independent detectors: with and without leakage protection (LP) and Pauli frame randomization (PFR).
		The marker is used to denote the input logical eigenstate ($\rho \propto X,Y,Z$).
		When LP is used, correlations between detectors are suppressed, measured through the conditional probability of a detector outcome given previous detector outcomes.
		(b) The same procedure is performed for H1-1E, and we additionally compare disabling leakage within the error model, which removes correlations.
		Experimental details given in Table \ref{tab:exp-settings}.
	}\label{fig:marginal-prob-detectors}
\end{figure*}

In this section we give numerical and experimental results for LSD-DRT and noise diagnostics, implemented using Quantinuum H1-1, a 20-qubit trapped-ion quantum computer with all-to-all connectivity.
We also make use of Quantinuum H1-1 emulator (H1-1E) -- a state vector simulated version of H1-1 with a full error model \cite{ryan2021realization, ryan2018quantum}.

\subsection{$\nkd{2}{1}{1}$ on H1-1}\label{sec:results-2,1,1}
We consider the smallest possible stabilizer code with non-trivial structure -- a $\nkd{2}{1}{1}$ CSS code.
We choose the code space to be stabilized by $-ZZ$ and with logical operators $\bar{X} = XX$ and $\bar{ZI}$.
This gives a logical subspace of $\text{span}(\ket{01}, \ket{10})$, otherwise called a Decoherence Free Subspace (DFS) \cite{kielpinski2002architecture}, for which any logical state is invariant under the $Z$ rotation $e^{i \theta Z} \otimes e^{i \theta Z}$.
Therefore logical states should be protected against coherent memory errors which take this form.
The use of this $\nkd{2}{1}{1}$ code therefore illustrates the broad applicability of these techniques even in regimes where one might expect some resistance to non-Pauli behaviour.

We first perform the noise diagnostic tests given in the previous section before proceeding with our main characterization protocol. 
Experimental details are given in Table \ref{tab:exp-settings}.

\subsubsection{Detector rate stabilization}\label{sec:noise-diag-experiments-stab}

In Figure~\ref{fig:leak-increase-click-rate}  we plot the click rate $p(\d_j = 1)$ of each detector in the sequence.
Without leakage protection, the click rate can be seen to increase with the length of the sequence, regardless of randomization over the Pauli frame.  
Whereas when using LP the click rate remains stable, suggesting error accumulation is localized.
Secondly, we see that -- both with and without leakage protection -- a randomization over Pauli frame brings the click rate for different eigenstates closer together. 
Recall that the click rate is independent of input states for pure Pauli noise.

\subsubsection{Detector pair correlations} \label{sec:noise-diag-experiments-corr}

For this flag scheme, when using leakage protection, after every 3 gadgets we have read out every physical qubit (see Figure~\ref{fig:lp-fault-tolerant-211-gadget}). 
Therefore when considering sequences of independent detectors, $\d_1,\ldots, \d_r$, any leakage may induce correlations between $\d_i$ and $\d_{i+1}$ but not between $\d_i$ and $\d_{i+2}$ (see Figure~\ref{fig:leakage-in-2,1,1}).

We test this by estimating the marginal probability $p(\d_{i+n}=1 | \d_{i}=1)$ of future detector $\d_{i+n}$ clicking \textit{given} a specific previous detector $\d_i$ has clicked. 
For pure Pauli noise, we would expect no corrections and therefore $p(\d_{i+n}=1 | \d_{i}=1) = p(\d_{i}=1)$.
As we consider independent pairs of detectors in sequence, measurement and hook errors do not contribute, leaving coherent effects and leakage errors. 

In Figure~\ref{fig:marginal-prob-detectors}, we plot the estimated marginal probabilities for $n=(1,2,3)$. 
Without leakage protection we see statistically significant higher marginal rates, both with and without Pauli frame randomization. 
When leakage protection is applied we see a large reduction in all marginal rates, with $p(\d_{i+2}=1 | \d_{i}=1)$ not being statistically different from the base rate $p(\d_{i}=1)$, consistent with the fact all physical qubits involved in $\d_i$ have been read out before $\d_{i+2}$.

These results suggest leakage protection is highly effective at converting correlating leakage errors into stochastic events.

\subsubsection{Benchmarking of flag-based syndrome measurements }

Through the LSD-DRT protocol, we learn the set of detector error channels, $\E^{\d}$, which are physical Pauli channels up to stabilizers of the code.
For any code, we can write these channels in a compact way as the sum of logical channels with each pure error $E \in \T$, such that
\begin{equation}
	\E^{\d}(\rho) = \sum_{L\in \gL, E \in \T}  p(L E|\d) EL \rho LE =  \sum_{E \in \T} E \ \E^\d_{E} (\rho) \ E
\end{equation}
for the (sub-normalized) logical Pauli channels $\E^\d_{E} (\rho) = \sum_{L\in \gL}  p(L E|\d) L\rho L$ where $\gL$ are a minimal set of logical operators for the code.
These channels are also the set $\{ \E^\d_{E} \}$ one would obtain following the detector with the ideal decoder.

\begin{figure*}[t]
	\begin{subfigure}{0.49\textwidth}
		\centering
		\includegraphics[height=7cm]{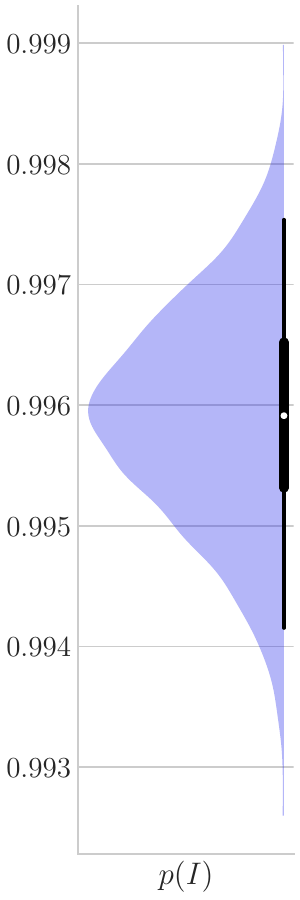}
		\includegraphics[height=7cm]{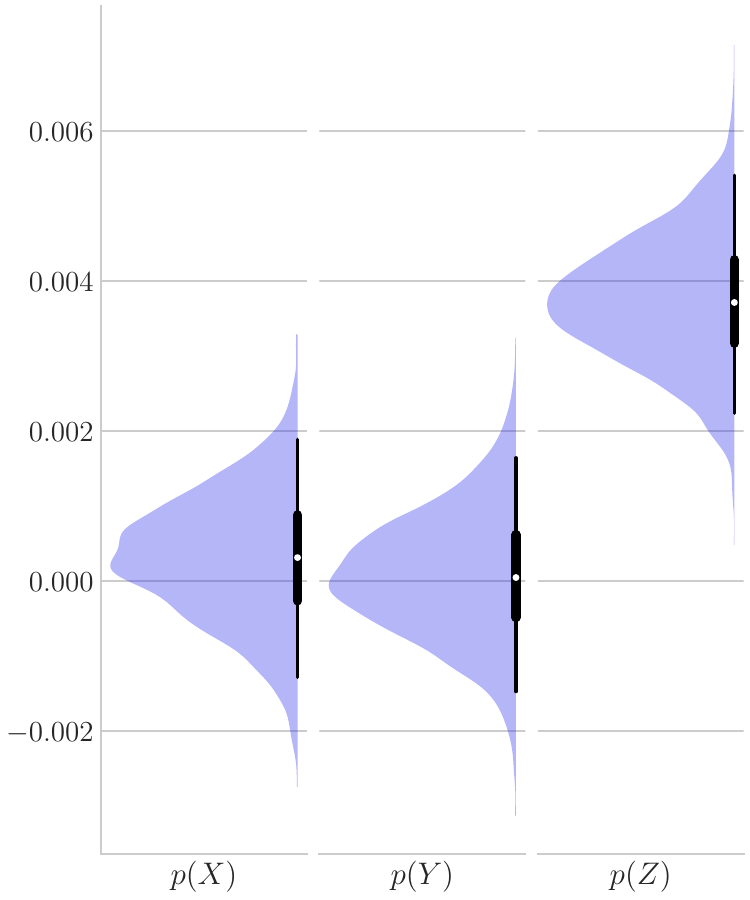}
		\caption{ Logical channel for detector outcome $\d=0$ when post-selected on no detectable error $\E_I^0/\tr(\E_I^0)$ }
		\label{fig:E_I^0}
	\end{subfigure}
	\begin{subfigure}{0.49\textwidth}
		\includegraphics[height=7cm]{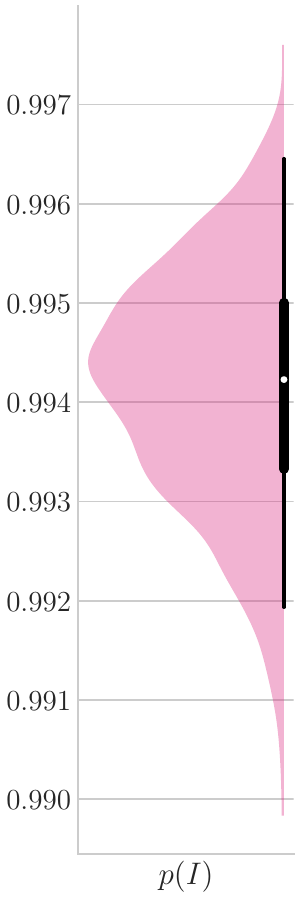}
		\includegraphics[height=7cm]{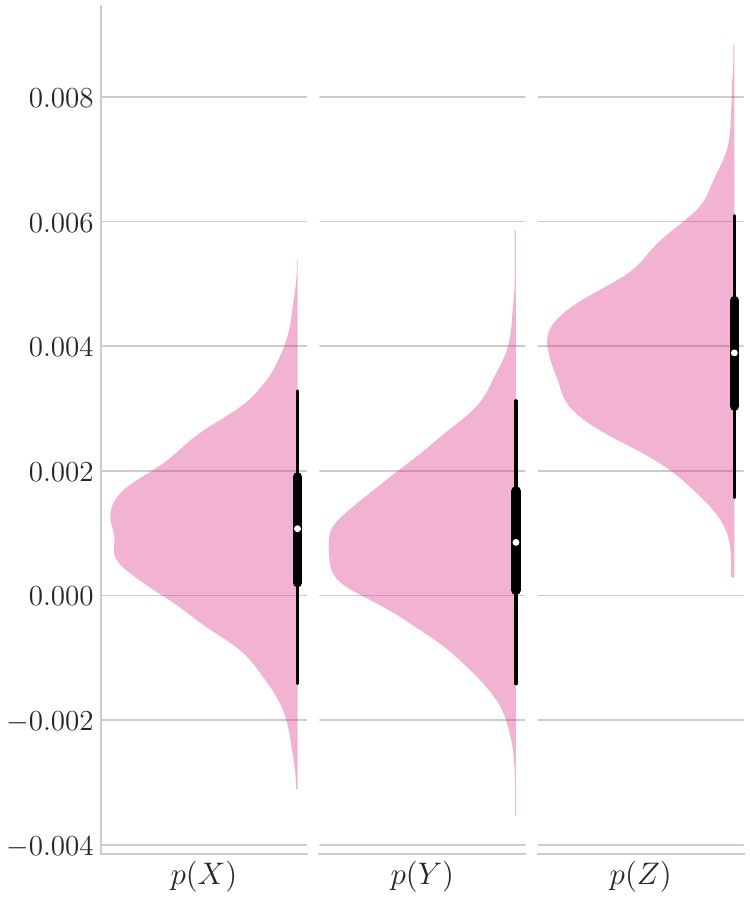}
		\caption{ Logical channel for detector outcome $\d=0$ when followed by an ideal decoder $(\E_I^0 + \E_X^0)/\tr(\E_I^0 + \E_X^0)$ }
		\label{fig:E_I^0+E_X^0}
	\end{subfigure}
	
	\caption{\textbf{Logical detector channel estimation for the $\nkd{2}{1}{1}$ code on H1-1 using Bayesian statistics.} 
	Violin plots showing syndrome-resolved, marginal posterior probability distributions of errors from a region of two syndrome gadgets with detector outcome $\d = 0$ that occurs with probability $p_{\text{H1-1}}(\d=0) = 0.9896 \pm 0.0002$ per detector.
	Data was obtained with the LSD-DRT protocol, using uniform priors for each eigenvalue.
	Experiment details are available in Table \ref{tab:exp-settings} and full syndrome-resolved posterior distributions for both detector outcomes in Figure~\ref{fig:LSD-2,1,1-bayes-uniform-prior}.  }
	\label{fig:bound-channels-2,1,1-bayes}
\end{figure*}

In Figure~\ref{fig:LSD-2,1,1-bayes}, we show the full results of LSD-DRT, including the average channel per gadget  -- where syndrome data is discarded.
We perform Bayesian post-processing on experimental data using both uninformative priors and informative priors (see Section~\ref{sec:bayes}), to produce marginal posterior distributions for each probability.

For both $\E^0$ and $\E^1$ we observe that $\E^\d_{I}$ are approximately logical dephasing channels where $p(I) \gg p(Z)>0$ while $p(X),p(Y) \approx 0$. 
The channels $\E^\d_{X}$ are closer to completely depolarizing, with $p(Z) \sim 3 p(X \lor Y)$. 

In Figure~\ref{fig:bound-channels-2,1,1-bayes}, we plot $\E^{0}_I$ and $\E^0_I + \E^0_X$ with re-normalization.
For post-selection, performance of the characterized detector region will be upper bounded by the case when any remaining physical error, $\E^{0}_X$,  will be caught within the next round of syndrome extraction.
This gives the channel $\E^{0}_I/\tr(\E^{0}_I)$ in Figure~\ref{fig:E_I^0}.

We can also consider performing corrections.
Within a circuit, a simple decoding strategy \cite{huang2024comparing} would be to apply no correction given $s_1 = s_2 = 0$ and correct using the pure error $IX$ given $s_1 = s_2 = 1$. 
As both of these outcomes give $\d = s_1 \oplus s_2 = 0$ performance will be upper-bounded by the $\E^0$ channel, and specifically $(\E^0_I + \E^0_X)/\tr(\E^0_I + \E^0_X)$.
Equivalently, this channel places a lower bound on the performance of post-selection, where the residual physical pure error is not caught and the logical error rate becomes averaged over both logical co-sets.

Finally, we note that the characterization given using LSD-DRT is complete -- in the sense that it can be used to calculate error channels for \textit{arbitrary sequences} of overlapping detectors as discussed in Figure~\ref{fig:fid-bounds-2-1-1}.

\subsection{ $\nkd{4}{2}{2}$ on H1-1E}

To demonstrate the generality of our approach, using H1-1E, we perform LSD-DRT for the $\nkd{4}{2}{2}$ code using a flag-based syndrome extraction scheme and the noise tailoring techniques in Section~\ref{sec:noisediagnostics}.
In Figure~\ref{fig:LSD-4,2,2-frequentist} we plot the results using the frequentist post-processing methods.

We compare the results of LSD-DRT with numerics for a simplified circuit-based noise model.
We use a characterization of the CNOT gate on H1-1, in the form of a two-qubit Pauli channel using parameters obtained from cycle benchmarking \cite{erhard2019characterizing} for the native maximally entangling gate, $e^{-i \frac{\pi}{4} Z \otimes Z}$.
We propagate this model through the syndrome extraction circuits given in Figure~\ref{fig:422-gadgets} to calculate detector channels isolated from other errors sources and finite sampling effects.

When compared to our direct approach, inferring syndrome-resolved channels using the CNOT noise model results in notable discrepancies in the distribution of particular errors but similar total error rates.
This suggests other error sources beyond a circuit level noise model significantly contribute to logical performance, highlighting the need for direct LSD type methods.

\section{Conclusion}

Developing logical noise characterization tools and diagnostics capable of testing the fundamental limits of practical hardware implementation of QEC is crucial for de-risking the path towards fault-tolerant computations at scale.
In this work, we introduced and demonstrated on hardware a SPAM-robust protocol to directly estimate logical Pauli channels conditioned on the detector outcomes that arise from a region of consecutive syndrome extraction circuits. 
This method provides access to syndrome-resolved logical-level noise structure beyond average error rates.
Our implementation focused on small error-detection codes, as these may be used as a first level in concatenated schemes \cite{dasu2025order, sommers2025observation}. 

Our protocol also applies to error-correcting codes, including a recovery operation.
It is naturally suitable for characterizing flag-based syndrome extraction where corrections depend on detectors -- parities of outcomes in two or more repeated stabilizer measurements.
The main theoretical limitation is that in its current form, the method cannot estimate the conditional error channels of a single syndrome extraction gadget with a trusted outcome -- only pairs of gadgets with a fixed detector.
For example, this is the case for individual stabilizer measurements based on Knill or Steane schemes.
Incorporating recent progress on benchmarking mid-circuit measurements and instruments \cite{zhang2025generalized, mclaren2025benchmarking, hines2025pauli} could offer a way forward, but raises further questions if all syndrome-dependent logical errors can be learned without enforcing additional assumptions on the noise models. 

An immediate extension of this protocol is to apply it to extended QEC regions including syndrome extraction, logical Clifford gates and corrections in order to determine the logical Pauli error channels (and pure errors) associated with specific detector outcomes.
Because our protocol actively uses the classical outcomes from independent detectors, this breaks down the degeneracy in estimating fidelities of Pauli errors equivalent under conjugation by the Clifford gates.
As a result, we could estimate the individual probabilities that are otherwise un-learnable using cycle benchmarking \cite{erhard2019characterizing}.
Our proofs directly apply to single code-blocks of Hermitian codes for which Clifford gates leave the stabilizer group invariant, but we leave generalizations for future work.

For larger quantum error correcting codes, certain syndromes correspond to rare events with insufficient statistics to fully resolve their associated channels. In such cases, one can still use our methods by coarse-graining syndromes with a probability of occurrence below a certain level to characterize their resulting average channel. Obtaining estimates with relative precision is necessary for finer syndrome-resolved error characterization, but remains challenging. In follow-on work we will aim to improve our Bayesian model and explore the use of ratio estimators which have been used to obtain relative precision when channels are purely Pauli \cite{flammia2020efficient} and the decay model does not require an off-set.

An important direction of future work is to compare the conditional error channels directly estimated at the logical level to those predicted from physical-level characterizations.
SPAM-robust techniques such as averaged circuit eigenvalue sampling \cite{flammia2021averaged} and extensions of cycle benchmarking can estimate the effective physical noise of Clifford circuits, and have become important tools for calibrating and understanding device performance.
Recent work \cite{hockings2025scalable} showed that ACES can scalably learn (contextual) circuit-level noise model in layers of Clifford gates that make up syndrome extraction circuits.
While, for hardware with a fixed architecture, characterizing layers under the assumption of localized error models may be more well-motivated, dynamic architectures such as QCCDs require more careful consideration that would be interesting to explore further and verify using a logical level benchmarking that relies on fewer assumptions.  

Finally, beyond validating the performance of quantum error correction on real hardware and testing failure mechanisms affecting scalability, benchmarking logical components could also be an important tool to improve decoding and reduce logical error rates. 
A key open question is to what degree noise-aware decoding can tolerate mismatches between modelled and actual hardware noise; whether simple circuit-level models suffice, or if logical-level priors are needed for the lower levels of a concatenated code.

\section{Acknowledgements}
We thank Karl Mayer and Shival Dasu for reviewing this manuscript. 

\bibliography{refs.bib}

\newpage
\onecolumngrid
\appendix

\begin{table}[h]
\begin{tabular}{@{}llllllllll@{}}
\toprule
                                                & Date (2025)                          & Device                          & Code                                 & Input states                    & Sequence length(s)                     & Copies                          & Shots                                       & PFR                             & LP         \\ \midrule
\multicolumn{1}{l|}{\textbf{Noise Diagnostics}} & \multicolumn{1}{l|}{March 31st}      & \multicolumn{1}{l|}{H1-1}       & \multicolumn{1}{l|}{$\nkd{2}{1}{1}$} & \multicolumn{1}{l|}{3}          & \multicolumn{1}{l|}{30}              & \multicolumn{1}{l|}{6}          & \multicolumn{1}{l|}{250}                    & \multicolumn{1}{l|}{\cmark}     & \cmark     \\
\multicolumn{1}{l|}{}                           & \multicolumn{1}{l|}{\texttt{"}}      & \multicolumn{1}{l|}{\texttt{"}} & \multicolumn{1}{l|}{\texttt{"}}      & \multicolumn{1}{l|}{\texttt{"}} & \multicolumn{1}{l|}{\texttt{"}}      & \multicolumn{1}{l|}{\texttt{"}} & \multicolumn{1}{l|}{\texttt{"}}             & \multicolumn{1}{l|}{\xmark}     & \cmark     \\
\multicolumn{1}{l|}{}                           & \multicolumn{1}{l|}{April 17th}      & \multicolumn{1}{l|}{\texttt{"}} & \multicolumn{1}{l|}{\texttt{"}}      & \multicolumn{1}{l|}{\texttt{"}} & \multicolumn{1}{l|}{\texttt{"}}      & \multicolumn{1}{l|}{\texttt{"}} & \multicolumn{1}{l|}{\texttt{"}}             & \multicolumn{1}{l|}{\cmark}     & \xmark     \\
\multicolumn{1}{l|}{}                           & \multicolumn{1}{l|}{\texttt{"}}      & \multicolumn{1}{l|}{\texttt{"}} & \multicolumn{1}{l|}{\texttt{"}}      & \multicolumn{1}{l|}{\texttt{"}} & \multicolumn{1}{l|}{\texttt{"}}      & \multicolumn{1}{l|}{\texttt{"}} & \multicolumn{1}{l|}{\texttt{"}}             & \multicolumn{1}{l|}{\xmark}     & \xmark     \\
\multicolumn{1}{l|}{}                           & \multicolumn{1}{l|}{\texttt{"}}      & \multicolumn{1}{l|}{H1-1E}      & \multicolumn{1}{l|}{\texttt{"}}      & \multicolumn{1}{l|}{\texttt{"}} & \multicolumn{1}{l|}{\texttt{"}}      & \multicolumn{1}{l|}{\texttt{"}} & \multicolumn{1}{l|}{10000}             & \multicolumn{1}{l|}{\cmark}     & \cmark     \\
\multicolumn{1}{l|}{}                           & \multicolumn{1}{l|}{\texttt{"}}      & \multicolumn{1}{l|}{\texttt{"}} & \multicolumn{1}{l|}{\texttt{"}}      & \multicolumn{1}{l|}{\texttt{"}} & \multicolumn{1}{l|}{\texttt{"}}      & \multicolumn{1}{l|}{\texttt{"}} & \multicolumn{1}{l|}{\texttt{"}}             & \multicolumn{1}{l|}{\xmark}     & \cmark     \\
\multicolumn{1}{l|}{}                           & \multicolumn{1}{l|}{\texttt{"}}      & \multicolumn{1}{l|}{\texttt{"}} & \multicolumn{1}{l|}{\texttt{"}}      & \multicolumn{1}{l|}{\texttt{"}} & \multicolumn{1}{l|}{\texttt{"}}      & \multicolumn{1}{l|}{\texttt{"}} & \multicolumn{1}{l|}{\texttt{"}}             & \multicolumn{1}{l|}{\cmark}     & \xmark     \\
\multicolumn{1}{l|}{}                           & \multicolumn{1}{l|}{\texttt{"}}      & \multicolumn{1}{l|}{\texttt{"}} & \multicolumn{1}{l|}{\texttt{"}}      & \multicolumn{1}{l|}{\texttt{"}} & \multicolumn{1}{l|}{\texttt{"}}      & \multicolumn{1}{l|}{\texttt{"}} & \multicolumn{1}{l|}{\texttt{"}}             & \multicolumn{1}{l|}{\xmark}     & \xmark     \\
\multicolumn{1}{l|}{\textbf{LSD-DRT}}               & \multicolumn{1}{l|}{April 25th-27th} & \multicolumn{1}{l|}{H1-1}       & \multicolumn{1}{l|}{\texttt{"}}      & \multicolumn{1}{l|}{\texttt{"}} & \multicolumn{1}{l|}{[6,12,18,36,72]} & \multicolumn{1}{l|}{\texttt{"}} & \multicolumn{1}{l|}{[350,1500,350,350,350]} & \multicolumn{1}{l|}{\cmark}     & \cmark     \\
\multicolumn{1}{l|}{}                           & \multicolumn{1}{l|}{April 30th}      & \multicolumn{1}{l|}{H1-1E}      & \multicolumn{1}{l|}{\texttt{"}}      & \multicolumn{1}{l|}{\texttt{"}} & \multicolumn{1}{l|}{\texttt{"}}      & \multicolumn{1}{l|}{\texttt{"}} & \multicolumn{1}{l|}{\texttt{"}}             & \multicolumn{1}{l|}{\texttt{"}} & \texttt{"} \\
\multicolumn{1}{l|}{}                           & \multicolumn{1}{l|}{May 2nd-4th}     & \multicolumn{1}{l|}{\texttt{"}} & \multicolumn{1}{l|}{$\nkd{4}{2}{2}$} & \multicolumn{1}{l|}{9}          & \multicolumn{1}{l|}{[12,36,72,144]}  & \multicolumn{1}{l|}{1}          & \multicolumn{1}{l|}{10000}                  & \multicolumn{1}{l|}{\texttt{"}} & \texttt{"} \\ \bottomrule
\end{tabular}
\caption{\textbf{Datasets.} Details of the experiments performed for this work using Quantinuum H1-1 and Quantinuum H1-1 emulator (H1-1E).
For efficiency, on H1-1 we use multiple independent copies of code in parallel, ensuring transport costs are minimized.
Sequence lengths are chosen to approximately minimize variances (see Appendix~\ref{sec:sample-complexity}) given physical error rates and machine time budget.
When Pauli Frame Randomization (PFR) is used we randomize over 10 sub-sequences.}
\label{tab:exp-settings}
\end{table}

\section{Fitting model for sequences of independent detectors under a Pauli noise}
\label{ap:mainprotocol}

Given an input state $\rho$ on the data qubits, then the action of a QEC gadget with the physical Pauli noise $\P$ acting on ancillas and data qubits just before destructively measuring the ancillas in the computational basis can be described via a quantum instrument $\{ \mathcal{I}_{\s} , {\s} \}$.
In this way, we account for both the classical outcomes (i.e. syndromes) and the post-measurement state.
Specifically, we can write 
\begin{align}
	\mathcal{I}_{\s }(\rho)  = \sum_{\m} \P^{\m} (\Pi_{\m \oplus \s} \rho  \Pi_{\m\oplus \s}),
\end{align}
where in this notation $\m$ ranges over all possible bitstrings of length equal to the number of measured stabilizers.
The $\P^{\m}$ is the reduced Pauli error channel acting only on the data qubits which arises from those Pauli errors in $\P$ that flip the outcome according to $\m$. 
More concretely, we write $\P (\cdot) = \sum_{P_a, P_d} p(P_{a} \otimes P_{d}) P_a \otimes P_d ( \cdot ) P_a \otimes P_d $ for some probability distribution $p$ over Pauli errors on ancilla and data qubits which are respectively labelled by $a$ and $d$.
Then, the reduced channel will be written as $\P^{\m} (\rho)= \sum_{P_d}  \left( \sum_{P_a : | \< \m |P_a |{\bf{0}}\> |=1} p(P_a\otimes P_d) \right) P_d \, \rho \,P_d$.
The probability of observing an outcome $\s$ is given by $\textrm{prob}(\s) = \Tr( \I_{\s}(\rho))$ and the post-measurement (normalized) state will be $ \rho \rightarrow \frac{\I_{\s}(\rho)}{\textrm{prob}(\s)}$. 

Now, for a sequence of $r$ QEC gadgets then we have that the post-measurement state once outcomes $\s^1,\ldots, \s^r$ have been obtained will be given by
\begin{align}
\rho \rightarrow \frac{1}{\textrm{prob}(\s^1,\ldots,\s^r)} \mathcal{I}_{\s^r} \circ \ldots \circ \I_{\s^1}(\rho).
\end{align}

\begin{figure*}[t]
\scalebox{0.8}
{\begin{quantikz} 
		& \setwiretype{n} &\gategroup[wires=4,steps=14,style={solid,rounded corners,fill=black!05,inner ysep=20pt, inner xsep=7pt},background,label style={label position=below,anchor=
			north,yshift=-0.2cm}]{{\sc DETECTOR ${\bf{D}} = {\bf{s}}^1 \oplus {\bf{s}}^2 $}}
		\gategroup[4,steps=5, style={dashed,rounded
			corners,fill=blue!10, inner 
			xsep=2pt},background,label style={label
			position=below,anchor=north,yshift=-0.2cm}]{{\sc
				QEC Gadget $\rightarrow {\bf{s}}^1 = (s_1^1, s_2^1)$ }}&  \lstick{$|+\>$} &\ctrl{2} \qw & \qw   &  \gate{H} \qw & \gate[4]{\mathcal{P}} \qw & \meter{s_1^{1}} \qw &  
		\gategroup[4,steps=5, style={dashed,rounded
			corners,fill=blue!10, inner
			xsep=2pt},background,label style={label
			position=below,anchor=north,yshift=-0.2cm}]{{\sc
				QEC Gadget  $\rightarrow {\bf{s}}^2 = (s_1^2, s_2^2)$ }}&     \lstick{$|+\>$}  &\ctrl{2} \qw  & \qw   &  \gate{H} \qw  &  \gate[4]{\mathcal{P}} \qw &\meter{s_1^2} \qw
		\\		
		&\setwiretype{n}  & 
		& \lstick{$|+\>$}&\qw   & \ctrl{1} \qw & \gate{H} \qw & \qw& \meter{s_2^1} \qw & 
		&\lstick{$|+\>$}   &\qw    &  \ctrl{1} \qw & \gate{H} \qw &  \qw& \meter{s_2^2} \qw
		\\
		\lstick[2]{Data qubits}	\qw & \qw & \qw 
		&\qw& \gate[2]{S_1}&  \gate[2]{S_2}& \qw & \qw & \qw &\qw
		&  \qw &     \gate[2]{S_1} & \gate[2]{S_2} &\qw &\qw &  \qw &\qw &\qw
		\\
		&\qw &  \qw& \qw
		& \qw & \qw  &\qw & \qw & \qw & \qw &\qw  
		& \qw  & \qw & \qw  &\qw &\qw&\qw&\qw
	\end{quantikz}}
	\label{fig:repeat-stabilizermeasurement}
	\caption{\textbf{Detector arising from repeated flag-based syndrome measurements with fault-tolerant circuit gadgets.}
	Our physical error model assumes each gadget is subject to the same Pauli channel that acts on all physical qubits (including both data and ancilla).
	This is a weaker assumption than circuit-level noise.}
\end{figure*}
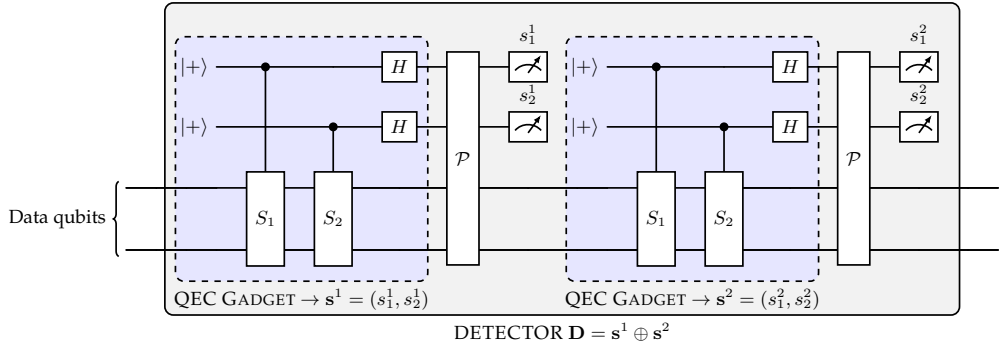

Expanding out in terms of the Pauli error channel and the projectors we get
\begin{align}
	\mathcal{I}_{\s^r} \circ \ldots \circ \I_{\s^1}(\rho)  & = \sum_{\m_1,\m_2\ldots,\m_r} \P^{\m_r} (\Pi_{\m_r \oplus \s^r}( \ldots \P^{\m_1} (\Pi_{\m_1\oplus \s^1}\rho \Pi_{\m_1\oplus \s^1})\ldots )\Pi_{\m_r \oplus \s^r} ), \nonumber \\
	& =  \sum_{\m_1,\m_2\ldots,\m_r} \P^{\m_r} \circ \P^{\m_{r-1}}_{\s^r\oplus\s^{r-1}\oplus \m_r\oplus \m_{r-1}} \circ \ldots \circ \P^{m_1}_{\s^1\oplus\s^2\oplus \m_1\oplus \m_2} (\Pi_{\m_1\oplus \s^1}\rho \Pi_{\m_1\oplus \s^1}),
	\label{eq:multiple-rounds-instrument}
\end{align}
where we used the subscript for the Pauli channels $\P^{\m}_{\e}$  to denote that from $\P^{\m}$ we only keep the contribution of those Pauli operators that have syndrome $\e$.
This follows from propagating the projectors through the Pauli error channel by making use of the fact that $\Pi_{\alpha} P \Pi_{\beta} = P \Pi_{\beta}$ if the Pauli operator $P$ has syndrome $\alpha\oplus \beta$  and zero otherwise.
Furthermore, since these errors are disjoint, summing over all possible syndromes gives the original channel such that $\P^{\m} = \sum_{\e} \P^{\m}_{\e}$.

In particular, the (un-normalized) state after two consecutive gadgets will be given by
\begin{align}
	\I_{\s^2}\circ\I_{\s^1}(\rho) = \sum_{\m_1,\m_2} \P^{\m_2} \circ \P^{\m_1}_{\s^1\oplus \s^2 \oplus \m_1\oplus \m_2}  (\Pi_{\m_1\oplus \s^1}\rho \Pi_{\m_1\oplus \s^1}).
\end{align}\\

We consider the output state given \emph{only} the information on these \emph{disjoint detectors}, as shown in Figure~\ref{fig:independent-detectors}.
This has the effect of discarding some classical information about the measurement outcomes; for instance, if both syndromes in a gadget agree, this will correspond to a trivial detector, and knowing only the latter means that we no longer have information about the individual syndrome outcomes themselves, only that they must be the same.
As a result, observing a particular detector means that a number of events, equal to the size of all possible syndrome outcomes, could have happened.
This marginalization over half of the syndrome outcomes leads to a decoupling of the information transmission through the repeated gadgets.
This way, errors that occur only on the ancillas, no longer propagate into the next detector. 

Forming detectors from disjoint pairs of syndrome measurement outcomes enables the characterization of the error channels associated with each detector.
Concretely, given that we \emph{only} observe the detector $\d$ from two rounds of syndrome measurements, then the (un-normalized) post-measurement state is
\begin{align}
	\E^{\d} (\rho): &= \sum_{\s^1} \mathcal{I}_{\s^1\oplus \d}\circ\mathcal{I}_{\s^1} (\rho) = \sum_{\m_1,\m_2,\s^1} \P^{\m_2}\circ\P^{\m_1}_{\d\oplus\m_1\oplus\m_2} (\Pi_{\m_1\oplus \s^1}\rho \Pi_{\m_1\oplus \s^1}), \\
	& = \sum_{\m_1,\m_2} \P^{\m_2}\circ\P^{\m_1}_{\d\oplus\m_1\oplus\m_2} (\rho) .
\end{align}

In a more formal way,  $\E = \{\E^{\d}, \d \}_{\forall \d}$ is an instrument over the detector outcomes and can be written as a channel from the input state to the output state and classical detector outcomes as $\E(\rho) = \sum_{\d} \E^{\d}(\rho) \otimes \ketbra{\d}$ where $\ket{\d}$ is a classical register. 

Two independent detectors $\d_1$ and $\d_2$ are described by
\begin{align}
\sum_{s^1, s^3} \I_{\d_2\oplus \s^3} \circ \I_{\s^3}\circ \I_{\d_1\oplus\s^1}\circ \I_{\s^1}(\rho) &= \hspace{-0.4cm} \sum_{\substack{\m_1,\ldots, \m_4, \\ \s^1, \s^3}} \P^{\m_4} \circ \P^{\m_3}_{\d_2\oplus \m_3\oplus \m_4} \circ \P^{\m_2}_{\d_1\oplus \s^1 \oplus \s^3 \oplus \m_3\oplus \m_2} \circ \P^{\m_1}_{\d_1\oplus \m_1\oplus \m_2} (\Pi_{\m_1\oplus \s^1}\rho \Pi_{\m_1\oplus \s^1}),\\
&= \hspace{-0.4cm} \sum_{\m_1,\ldots, \m_4} \P^{\m_4} \circ \P^{\m_3}_{\d_2\oplus \m_3\oplus \m_4} \circ \P^{\m_2} \circ \P^{\m_1}_{\d_1\oplus \m_1\oplus \m_2} (\rho),\\
&= \E^{\d_2} \circ \E^{\d_1}(\rho).
\end{align}
We note that, in the above, we use the fact that $\sum_{\s^3} \P^{\m_2}_{\d_1\oplus \s^1\oplus \s^3\oplus \m_3\oplus \m_2} = \P^{\m_2}$, since the sum is over all possible syndromes for the Pauli operators appearing in $\P^{m_2}$, so we are no longer restricted to parts of the channel.

The generalization to $2r$ repeated syndrome measurements follows similarly, in a recursive fashion leading to the following result
\begin{lemma}
	Given that detectors $\d_1,\ldots, \d_r$ are observed from independent syndrome pairs, then the (un-normalized) post-measurement state after the $2r$ rounds is given by 
	\begin{align}
		\sum_{\{\s^{2i -1} \}_{i=1}^{r}} \I_{\d_r \oplus \s^{2r-1}} \circ \I_{\s^{2r-1}} \circ \ldots \circ \I_{\d_1\oplus \s^1} \circ \I_{\s^1} (\rho)  = \E^{\d_r} \circ \ldots \circ \E^{\d_2} \circ \E^{\d_1} (\rho).
		\label{eq:product-independent-detectors}
	\end{align}
\label{lem:indpendent-channels}
\end{lemma}
\begin{proof}
The above follows from the following calculation, starting from the $\textrm{LHS}$ and expanding it out using eqn.~(\ref{eq:multiple-rounds-instrument})
\begin{align*}
	&= \textrm{LHS}, \\
	&=\hspace{-0.6cm}   \sum_{\substack{\s^1,\s^3,\ldots,\s^{2r-1}  \\  \m_1,\m_2\ldots,\m_{2r} }} \P^{\m_{2r}} \circ \P^{\m_{2r-1}}_{\d_r \oplus \m_{2r}\oplus \m_{2r-1}} \circ \P^{\m_{2r-2}}_{\s^{2r-1}\oplus \d_{r-1}\oplus \s^{2r-3}\oplus \m_{2r-1}\oplus \m_{2r-2}} \circ \ldots \circ \P^{m_1}_{\d_1\oplus \m_1\oplus \m_2} (\Pi_{\m_1\oplus \s^1}\rho \Pi_{\m_1\oplus \s^1}), \\
	&= \hspace{-0.6cm} \sum_{\substack{\s^1,\s^3,\ldots,\s^{2r-3}  \\  \m_1,\m_2\ldots,\m_{2r} }} \P^{\m_{2r}} \circ \P^{\m_{2r-1}}_{\d_r \oplus \m_{2r}\oplus \m_{2r-1}} \circ \left(\sum_{\s^{2r-1}}\P^{\m_{2r-2}}_{\s^{2r-1}\oplus \d_{r-1}\oplus \s^{2r-3}\oplus \m_{2r-1}\oplus \m_{2r-2}} \right) \circ \ldots \circ \P^{m_1}_{\d_1\oplus \m_1\oplus \m_2} (\Pi_{\m_1\oplus \s^1}\rho \Pi_{\m_1\oplus \s^1}), \\
	& =  \hspace{-0.6cm} \sum_{\substack{\s^1,\s^3,\ldots,\s^{2r-3}  \\  \m_1,\m_2\ldots,\m_{2r} }} \P^{\m_{2r}} \circ \P^{\m_{2r-1}}_{\d_l \oplus \m_{2r}\oplus \m_{2r-1}} \circ \P^{\m_{2r-2}} \circ \P^{\m_{2r-3}}_{\d_{r-1} \oplus \m_{2r-3}\oplus \m_{2r-2}} \circ \ldots \circ \P^{m_1}_{\d_1\oplus \m_1\oplus \m_2} (\Pi_{\m_1\oplus \s^1}\rho \Pi_{\m_1\oplus \s^1}), \\
	& = \E^{\d_{r}} \circ \left(  \sum_{\substack{\s^1,\s^3,\ldots,\s^{2r-3}  \\  \m_1,\m_2\ldots,\m_{2r-2} }}  \P^{\m_{2r-2}} \circ \P^{\m_{2r-3}}_{\d_{r-1} \oplus \m_{2r-3}\oplus \m_{2r-2}} \circ \ldots 
	\circ \P^{m_1}_{\d_1\oplus \m_1\oplus \m_2} (\Pi_{\m_1\oplus \s^1}\rho \Pi_{\m_1\oplus \s^1})\right).
\end{align*}
Now the expression in the bracket is the same as the post-measurement state for $2r-2$ rounds given only detectors $\d_1,\ldots, \d_{l-1}$, and therefore we can recursively apply the same calculation to obtain
\begin{align}
	\textrm{LHS} &= \E^{\d_r} \circ \ldots \circ \E^{\d_2} \sum_{\substack{\s^1 \\ \m_1,\m_2}} \P^{\m_2} \circ\P^{\m_1}_{\d_{1}\oplus \m_{1}\oplus \m_2} (\Pi_{\m_1\oplus \s^1} \rho \Pi_{\m_1\oplus \s^1} ), \\
	& =\E^{\d_r} \circ \ldots \circ \E^{\d_2}\circ\E^{\d_1} (\rho),
\end{align}
which completes the proof of eqn.~(\ref{eq:product-independent-detectors}).
\end{proof}
\begin{lemma}
	Suppose that we prepare input state $\rho$, apply $2r$ repeated rounds to obtain detectors $\d_1,\ldots, \d_r$ from independent consecutive syndrome pairs, and measure a Pauli operator $Q$.
	Then we have that
	\begin{align}
		\mathbb{E}[Q| \d_1,\ldots,\d_{r}, \rho ]  = \prod_{\d} \frac{\lambda_{\d}^{n_{\d}}(Q) }{\lambda_{\d}^{n_{\d}}(\iden)} \Tr(Q\rho),
	\end{align}
where $\d$ ranges over all possible error syndromes and $n_{\d}$ represents the number of times a syndrome $\d$ has  occurred in the sequence of detectors $\d_1,\ldots,\d_r$.
\end{lemma}
\begin{proof}
	The expected value of a Pauli operator $Q$, given that the initial state is $\rho$ and detectors $\d_1,\ldots,\d_r$ were measured is 
	\begin{align}
		\mathbb{E}(Q|\d_1,\ldots,\d_r, \rho) = \Tr(Q \rho|_{\d_1,\ldots\d_r}),
	\end{align}
where $\rho|_{\d_1,\ldots\d_r}$ is the post-measurement state after observing detectors $\d_1\ldots\d_r$.
From Lemma~\ref{lem:indpendent-channels} we have that the post-measurement state is simply
\begin{align}
	\rho|_{\d_1,\ldots,\d_r} = \frac{\E_{\d_r} \ldots \E_{\d_1} (\rho)}{\Tr( \E_{\d_r} \ldots \E_{\d_1} (\rho) )}.
\end{align}
Therefore
\begin{align}
	\mathbb{E}(Q|\d_1,\ldots,\d_r, \rho) = \frac{\Tr(Q \E_{\d_r} \ldots \E_{\d_1} (\rho) )}{\Tr( \E_{\d_r} \ldots \E_{\d_1} (\rho) )} =  \frac{\lambda_{\d_r}(Q) \, \ldots \, \lambda_{\d_1}(Q) \Tr(Q\rho)}{\lambda_{\d_r}({\iden}) \ldots \lambda_{\d_1}(\iden)},
\end{align}
where the last equality follows from cyclicity of trace and the fact that detector error channels are Pauli so that $\E^{\d} (Q) = \lambda_{\d}(Q) Q$.
\end{proof}

\begin{lemma}
	Given a Pauli channel, suppose that we can determine the eigenvalues of the channel $\lambda(Q)$ wherever $Q\in \gL\cdot \gS$, that is for all the Pauli operators that commute with the stabilizer group.
	Then, we can determine the error probabilities up to a stabilizer, namely
	\begin{align}
		p(P\S) := \frac{1}{|\S|}  \sum_{S\in \S} p(PS),
	\end{align}
	for all Pauli operators $P$. 
	\label{lem:eigenvalueslogical}
\end{lemma}
\begin{proof}
	Given a fixed Pauli operator $P$, the Walsh-Hadamard transform relating probabilities and eigenvalues of an error channel implies that, for every $S$, 
	\begin{align}
		p(PS) = \frac{1}{4^n}\sum_{Q\in \gP_n} \lambda(Q) (-1)^{\omega(PS, Q)}.
    \end{align}
	Marginalizing over all Pauli operators in the stabilizer group results in	
	\begin{align}
		p(P\S) &= \frac{1}{4^n 2^{n-k}} \sum_{\substack{Q\in \gP_n,\\ S\in\S}} \lambda(Q) (-1)^{\omega(PS,Q)}, \\
		&  =  \frac{1}{4^n 2^{n-k}} \sum_{\substack{Q\in \gP_n }} \lambda(Q) (-1)^{\omega(P,Q)}  \left(\sum_{S\in \S} (-1)^{\omega(S,Q)}\right).
	\end{align}
	Now, the RHS bracket above is non-zero only for elements of the normalizer of the stabilizer group, which is defined as $\N(\S) = \{ Q\in \gP_n :  [Q, S] = 0 \  \   \forall S\in \S\}$.
	Specifically,
	\begin{align} \label{eq:normalizerproj}
		 \sum_{S\in \S} (-1)^{\omega(S,Q)} = |\S| \delta_{Q \in \N(\S) }.
	\end{align}
	This is because, if $\S= \<S_1,\ldots, S_{n-k}\>$, then any element $S\in \S$ will be given by $S  = S^{a_1}_1 \ldots  S_{n-k}^{a_{n-k}}$ for $a_i \in \{0,1\}$.
	Suppose that $Q$ has error syndrome ${\bf{e}}$.
	If the dot product modulo 2 is given by ${\bf{e}} \cdot {\bf{a}} = e_1 a_1 \oplus e_2 a_2 \ldots\oplus e_{n-k} a_{n-k} $, then $(-1)^{\omega(S,Q)} = (-1)^{{\bf{e}\cdot a}} $.
	Summing over all ${\bf{a}} \in \{0,1\}^{n-k}$, we get either $2^{n-k}$ if $\bf{e}$ is the zero vector or zero otherwise.
	This means $\sum_{S\in \S} (-1)^{\omega(S,Q)} = \sum_{\bf{a}} (-1)^{\bf{e}\cdot a} = \delta_{\bf{e},0} |\S|$.
	Furthermore, the Pauli operators with trivial syndrome correspond exactly to those that commute with all elements of the stabilizer group, hence we obtain eqn.~(\ref{eq:normalizerproj}).
	Finally,
	\begin{align}
		p(P\S) = \frac{1}{4^n} \sum_{Q\in \N(\S)} \lambda(Q) (-1)^{\omega(P,Q)}.
	\end{align}
\end{proof}

\section{Non-Pauli behaviour}
\subsection{Unitary errors within syndrome extraction affecting detector click rates}\label{sec:non-pauli-effects}
Consider a QEC gadget that measures a stabilizer $S$ on the state $\F(\rho)$ where $\F$ is some error channel, and $\rho$ is a logical state stabilized by $S$.
The probability of a non-trivial outcome to the gadget will therefore be proportional to the expectation value $\tr[S \F(\rho)]$. 
If $\F$ is simply a Pauli channel then we have $\tr[S \F(\rho)] = \lambda_{\F}(S) \tr[S \rho] = \lambda_{\F}(S) \tr[\rho] = \lambda_{\F}(S)$, as all logical states $\rho$ are invariant under the stabilizer. 
Therefore, for Pauli noise, there is no state dependence in the probability of a non-trivial syndrome. 

However, now consider the case where $\F$ contains some unitary perturbation such that $\F(\rho) := (1-\eta) \P(\rho) + \eta U \rho U^\dagger$ for a Pauli channel $\P$, unitary operation $U$ and small constant $\eta$. 
We have $\tr[S \F(\rho)] = (1-\eta) \lambda_{\P}(S) + \eta \tr[S U \rho U^\dagger ] $, however (through the cyclical property of the trace) the perturbative term is proportional to another observable $ \tr[U^\dagger S U \rho ]$. 
If $U^\dagger S U$ is outside of the stabilizer group then $\tr[S \F(\rho)]$ will have a state dependence proportional to $\eta$.

Putting this together, statistically significant differences in the probability of a non-trivial syndromes for a QEC gadget between logical input states are a witness to non-Pauli behaviour. Here we have used unitary errors,  but the argument is similar for other state-dependent forms of error, such as leakage.

\subsection{Leakage errors within state preparation and measurement}
Within both post-processing methods we add a small offset to the fit derived in eqn.~(\ref{eqn:analytical-fit-function}), to account for non-Pauli effects within state preparation and measurement.
We now show why this parameter may be required using a simplified leakage model.
Consider a leakage channel $\F(\rho) := (1-\eta) \rho + \eta \ \tr[\rho] \ketbra{2}$ for small constant $\eta$ where $\ketbra{2}$ is state outside of the computational sub-space occupied by $\rho$.
Typically (within a trapped-ion device) measurement of a leaked qubit returns a fixed outcome, so for any given observable $Q$ we write $\tilde{Q} = Q \oplus \ketbra{2}$ for the generalized operator such that $\tr[\tilde{Q} \ketbra{2}] = 1$.
Putting this together, we have $\tr[\tilde{Q} \F(\rho)] = (1-\eta) \tr[Q \rho] + \eta$ giving an component independent of $\rho$.
Therefore, given this leakage occurs within final basis rotations and measurement, we could expect a small offset in any observable, independent of the previous state of the system.
To account for this, we add an additional constant SPAM parameter to the fit in eqn.~(\ref{eqn:least-squares-fit}) and eqn.~(\ref{eqn:bayes-fit}).

\section{Sample complexity and data processing with a weighted linear regression under a Pauli noise model}\label{sec:sample-complexity}

For a Pauli noise model, we have shown that the expectation value of a Pauli operator in the normalizer $Q\in \mathbb{L}\cdot \mathbb{S}$ given a sequence of independent detector outcomes is given by the following fitting model 
\begin{align} \mathbb{E}(Q|_{\vec{n}_{\d}})  = A_{Q} \prod_{\d}\lambda_{\d}^{n_{\d}}(Q).
\label{eq:fit}
\end{align}

Experimentally, in each shot we obtain a sequence of syndromes from which we calculate the detector counts  $\vec{n}_{\d}$ and a measurement outcome $q^{j} \in \{0,1\}$ of the observable $Q$.
The row vector $\vec{n}_{\d}$ is indexed by all the possible detector values $\d$ with each entry being the number of times it occurred $n_\d$.
If all syndromes of an $[[n,k,d]]$ code are measured by the gadget we want to characterize, then the vector's length will be $2^{n-k}$.
Furthermore, the row vector $\vec{n}_{\d}$ can be seen as a random variable that follows a multinomial distribution with parameters including $r$ (the total number of detectors), and $\vec{p}(\d)$, where the latter corresponds to the vector of probabilities $p(\d)$ for obtaining a detector of value $\d$, for all possible outcomes.
The probability that we obtain a vector of detector counts with values $\vec{n}_{\d}$ is given by
\begin{align}
\mathbb{P}( \vec{n}_{\d}) = \frac{r!}{\vec{n}_{\d}!} \prod_{\d} p^{n_{\d}}(\d).
\end{align}
Suppose that we take $N$ shots of final measurements of $Q$, for a sequence of $r$ detectors, with input state a $+1$ eigenstate of $Q$.
Denote by $N_{\vec{n}_\d}$ the number of shots in which a particular sequence $\vec{n}_{\d}$ occurs.
We then construct the sample mean estimators $\bar{q}|_{\vec{n}_{\d}} = \frac{1}{N_{\vec{n}_{\d}}} \sum_{j=1}^{N-{\vec{n}_{\d}} } q^{(j)}$ which is an unbiased estimator for the expectation value $\mathbb{E}(Q|_{\vec{n}_{\d}})$ with variance $\mathbb{V}( \bar{q}|_{\vec{n}_{\d}}) = \frac{\sigma_0}{N \mathbb{P}(\vec{n}_{\d})}$ where $\sigma_0$ is the single-shot variance for the Bernoulli RV $Q|_{\vec{n}_{\d}}$ so that $\sigma_0 = \left( 1 - \frac{\mathbb{E}(Q|_{\vec{n}_{\d}}) + 1}{2} \right) \frac{\mathbb{E}(Q|_{\vec{n}_{\d}}) +1}{2}$. 
We construct a matrix $\bf{M}$ with rows consisting of all the \emph{unique} detector counts vectors $\vec{n}_{\d}$ obtained from the experiment, including data obtained for different sequence lengths.
We append a single column with all entries equal to $1$. 

In a regime where we can take logarithms of eqn.~(\ref{eq:fit}), finding parameters of our fit model can be determined from a linear regression 
\begin{align}
	\bf{y} = \bf{M} \bf{x},
\end{align}
where $\bf{y}$ is a column vector obtained from the experimental data, with entries given by the logarithm of the sample mean estimator of the final measurement $\bar{q}|_{\vec{n}_{\d}}$ corresponding to the particular vector of detector counts in the row entries of $\bf{M}$.
We seek to find the solution given by the column vector $\bf{x}$  which corresponds to $ \left[\begin{array}{c} \log \boldsymbol{\lambda}_{\d} \\ \log(A_Q)\end{array} \right]$, with entries given by the logarithm of the eigenvalues $\lambda_{\d}(Q)$.

Formally, the weighted least squares will give the following estimator for $\bf{x}$:
\begin{align}
\hat{\bf{x}} = {(\bf{M} ^{T}W \bf{M})^{-1} \bf{M}^{T}} {\bf{ W}} \bf{y},
\end{align}
where $\bf{W}$ is a diagonal matrix with entries given by $\frac{1}{\mathbb{V}(\bf{y}_i)}$, the inverse variances. 
Using a first-order Taylor series approximation, $ \mathbb{V}[{\bf{y}}_i] = \mathbb{V}[\log(\bar{q}|_{\vec{n}_{\d}})] \approx \frac{\mathbb{V}[\bar{q}|_{\vec{n}_\d}]}{\mathbb{E}[Q|\vec{n}_{\d}]^2}$ the diagonal entries of $\bf{W}$ are given by $\frac{(\mathbb{E}[Q|\vec{n}_{\d}]^2) N\mathbb{P}(\vec{n}_{\d})}{\sigma_0} $. The variances on the eigenvalue estimators are given by diagonal elements of the covariance matrix,
\begin{align}
	\mathbb{V}(\bf{x}_{i})  = \left( \bf{M} ^{T}W \bf{M} \right)^{-1}_{ii}.
\end{align}

\section{Channel probabilities within post-processing methods}\label{sec:channel-probs-issues}
In both post-processing methods, we estimate a set of Pauli eigenvalues with confidence (or credible) intervals before performing a deterministic Fourier transform to obtain sets of probabilities.
In the realistic regime, generally, the eigenvalues will be close to one, and many probabilities will be close to zero. 
Therefore, due to finite sampling, we can end up with slightly negative means for probabilities when we perform the final inversion, though typically with uncertainty bars (credible intervals) overlapping with zero.
This is similar to the issue of catastrophic cancellation in numerical analysis \cite{benz2012dynamic}.

To remove this effect, for the frequentist approach, we could project our point estimates into the space of positive probabilities with a further least squares fit \cite{harper2020efficient}.
With Bayesian methods, we can build the estimated probabilities into the model at the highest level of the hierarchy in Figure~\ref{fig:bayes-model}.
As we aim to constrain the probabilities to be positive and to sum to one, the prior we must use is given by a \textit{Dirichlet} distribution.
However, such distributions cannot be used as non-informative priors, due to the summation constraint.
Therefore, we would restrict using this larger model to cases where an informative prior can be placed on probabilities; such as from device-level benchmarks or previously obtained estimates from LSD-DRT.

\section{Figures}

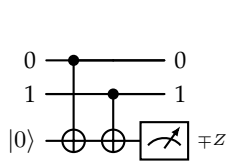
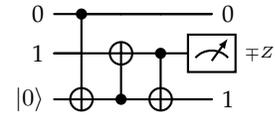
\begin{figure}[h]
    \begin{subfigure}{0.47\textwidth}
        \centering
        \begin{quantikz}[row sep=0.3cm, column sep = 0.2cm]
        \lstick{0  } & \ctrl{2} & \qw      & \rstick{0} \qw \\
        \lstick{1  } & \qw	  & \ctrl{1} & \rstick{1} \qw \\
        \lstick{ $\ket{0}$  } & \targ{} & \targ{} & \meter{} \rstick{$\scriptstyle \mp Z$}  \\
    	\end{quantikz}
            \caption{\textbf{Gadget of depth 2.}
            A weight-1 $X$ error before the final auxiliary gate propagates to a weight-1 $X$ error on the data qubits, which is caught by the next round of error detection.}
            \label{fig:fault-tolerant-211-gadget}
    \end{subfigure}
	\hspace{0.2cm}
    \begin{subfigure}{0.47\textwidth}
         \begin{quantikz}[row sep=0.2cm, column sep = 0.2cm]
        \lstick{0  }   & \ctrl{2}  &     \qw       &       \qw         & \rstick{0} \qw \\
        \lstick{1  }   &     \qw      & \targ{}    & \ctrl{1}     &  \meter{} \rstick{$\scriptstyle \mp Z$}  \\
        \lstick{ $\ket{0}$  }     & \targ{}   & \ctrl{-1}  & \targ{}          & \rstick{1} \qw \\
    	\end{quantikz}
            \caption{\textbf{Leakage-protected (LP) gadget of depth 3.}
            By swapping which physical qubits represent the  auxiliary and data qubits, we can systematically readout all physical qubits, reducing leakage errors to a constant level proportional to the depth of the gadget.}
            \label{fig:lp-fault-tolerant-211-gadget}
    \end{subfigure}
   \caption{\textbf{Fault-tolerant flag syndrome extraction gadgets for the $\nkd{2}{1}{1}$ code.}}
   \label{fig:211-gadgets}
\end{figure}

\begin{figure*}[h]

	\begin{subfigure}{0.90\textwidth}
        \centering
        \begin{quantikz}[row sep={1cm,between origins}, column sep=0.6cm]
		\lstick[brackets=none]{...} 
		& \gate[wires=2]{QEC} \gategroup[2,steps=2, style={dashed,rounded corners, inner ysep=2pt, color = black},background,label style={label position=below,anchor=north, yshift =-0.2cm, xshift=0.2cm, color = black}]{{\sc Detector $\d_i$}}  
		& \gate[wires=2]{QEC} \gategroup[2,steps=2, style={dashed,rounded corners, inner xsep=0pt, color = gray},background,label style={label position=above,anchor=north, yshift =0.2cm, color = gray}]{{\sc Detector $\d_{i+1}$}} 
		& \gate[wires=2]{QEC} \gategroup[2,steps=2, style={dashed,rounded corners, inner ysep=2pt, color = black},background,label style={label position=below,anchor=north, yshift =-0.2cm, xshift=0.2cm, color = black}]{{\sc Detector $\d_{i+2}$}} 
		& \gate[wires=2]{QEC}  \gategroup[2,steps=2, style={dashed,rounded corners, inner xsep=0pt, color = gray},background,label style={label position=above,anchor=north, yshift =0.2cm, color = gray}]{{\sc Detector $\d_{i+3}$}} 
		&  \gate[wires=2]{QEC}  \gategroup[2,steps=2, style={dashed,rounded corners, inner ysep=2pt, color = black},background,label style={label position=below,anchor=north, yshift =-0.2cm, xshift=0.2cm, color = black}]{{\sc Detector $\d_{i+4}$}} 
		& \gate[wires=2]{QEC} \gategroup[2,steps=2, style={dashed,rounded corners, inner xsep=0pt, color = gray},background,label style={label position=above,anchor=north, yshift =0.2cm, color = gray}]{{\sc Detector $\d_{i+5}$}} 
		&  \gate[wires=2]{QEC}  \gategroup[2,steps=2, style={dashed,rounded corners, inner ysep=2pt, color = black},background,label style={label position=below,anchor=north, yshift =-0.2cm, xshift=0.2cm, color = black}]{{\sc Detector $\d_{i+6}$}}
		& \gate[wires=2]{QEC}  
		& \qw \rstick[brackets=none]{...}\\
		\lstick[brackets=none]{...} & \qw  & \qw & \qw & \qw  &\qw &\qw &\qw  &\qw &\qw   \rstick[brackets=none]{...}
		\end{quantikz}
		\caption{ \textbf{Overlapping detectors within a sequence.} }
    \end{subfigure}

    \begin{subfigure}{0.80\textwidth}
		\includegraphics[width=0.80\textwidth]{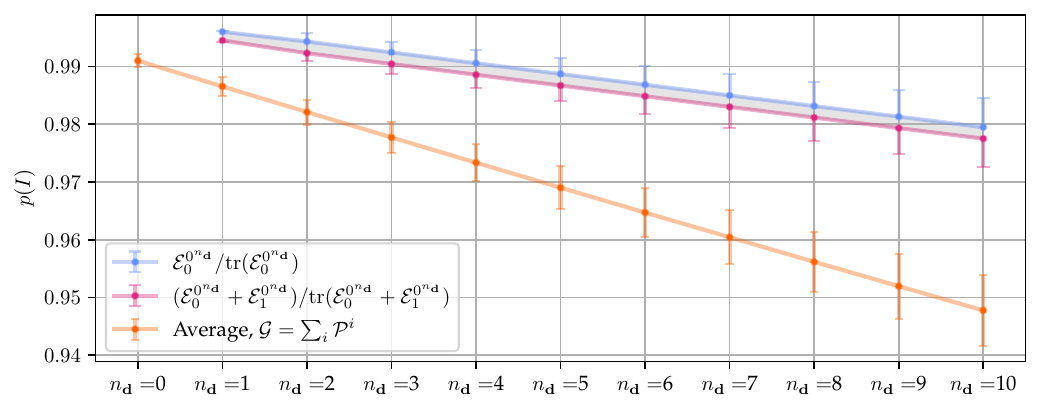}
   \caption{\textbf{Logical fidelity for sequences of overlapping detectors.} For length $n_{\d} =0$, a single gadget is performed.}
    \end{subfigure}
	\caption{ \textbf{Performance bounds on logical fidelity for post-selection using the $\nkd{2}{1}{1}$ code on H1-1. }
	Using the channels estimated in Figure~\ref{fig:bound-channels-2,1,1-bayes}, error channels $\E^{\d_0,...,\d_N}_E$ associated with finite sequences of overlapping detectors can be calculated numerically, where $\d_0,...,\d_N=0^{n_{\bf d}}$ indicates that all detectors had outcome zero (e.g. post-selection). }
	\label{fig:fid-bounds-2-1-1}
\end{figure*}

\tikzset{
leak/.style={circle,fill=ibm1, draw = ibm1, line width=0.01pt, inner sep = 0.0, minimum size = 0.03pt, outer sep = 0.0}
}

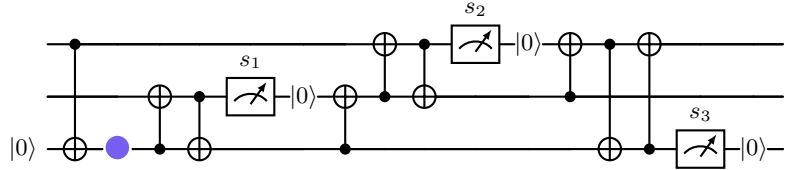
\begin{figure}[h]
	\centering
	\begin{quantikz}[row sep=0.2cm, column sep = 0.2cm]
          & \ctrl{2}  & \qw       & \qw       & \qw       & \qw 		& \qw 		& \qw		& \targ{} 	& \ctrl{1} 	& \meter{s_{2}}  & \ket{0} 	& \targ{} 	& \ctrl{2} 	& \targ{} 	& \qw 		& \qw 		& \qw 			\\
           &     \qw   & \qw       & \targ{}   & \ctrl{1}  & \meter{s_1}  	& \ket{0}  	& \targ{} 	& \ctrl{-1}	& \targ{} 	& \qw 		& \qw 		& \ctrl{-1}	& \qw 		& \qw 		& \qw 		& \qw 		& \qw 			\\
\lstick{$\ket{0}$  }   & \targ{}   &    \scaleto{\color{ibm1} \displaystyle  \bullet}{10pt} \qw    & \ctrl{-1} & \targ{}   & \qw 		& \qw 		& \ctrl{-1}	& \qw		& \qw		& \qw 		& \qw 		& \qw 		& \targ{} 	& \ctrl{-2} & \meter{s_{3}}  & \ket{0} 	& \qw 			\\
    	\end{quantikz}
	\caption{\textbf{Rounds of syndrome extraction with leakage protection for the $\nkd{2}{1}{1}$ code.}
	If a leakage event occurs at {\Large $\color{ibm1} \displaystyle  \bullet$} then -- in the simplest instance -- subsequent CNOT operations involving this qubit will not be applied, and measurement $s_{i+2}$ will return $1$.
	If independent detectors are created $\d_1 = s_1 \oplus s_2$, $\d_2 = s_3 \oplus s_4$, \dots, then leakage at {\Large $\color{ibm1} \displaystyle  \bullet$} can induce correlations between $\d_1$ and $\d_2$ but not between $\d_1$ and $\d_3$ and higher.
	}
	\label{fig:leakage-in-2,1,1}
\end{figure}

\begin{figure*}[h]
    \begin{subfigure}{0.49\textwidth}
        \centering
        \includegraphics[height=6cm]{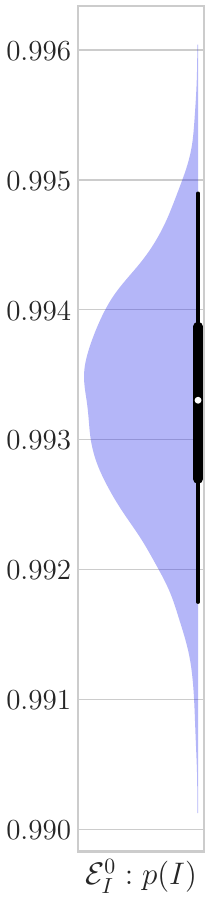}
		\includegraphics[height=6cm]{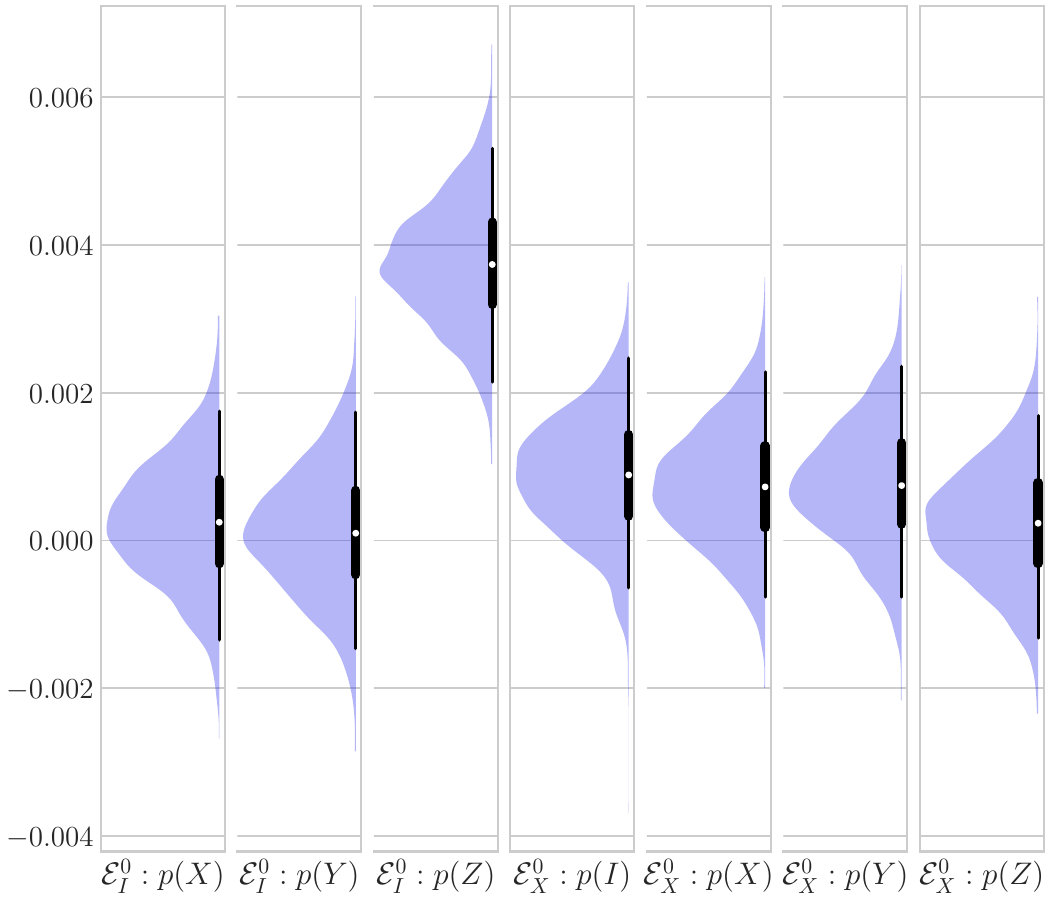}
		\caption{ \textbf{ $\E^{0}$} with $p_{\text{H1-1}}(\d=0) = 0.9896 \pm 0.0002$. }
    \end{subfigure}
	\hspace{0.1cm}
	\begin{subfigure}{0.49\textwidth}
		\includegraphics[height=6cm]{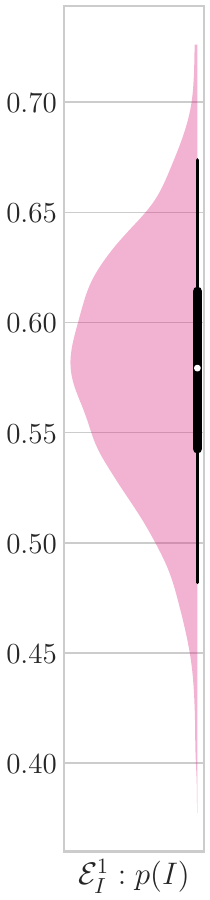}
		\includegraphics[height=6cm]{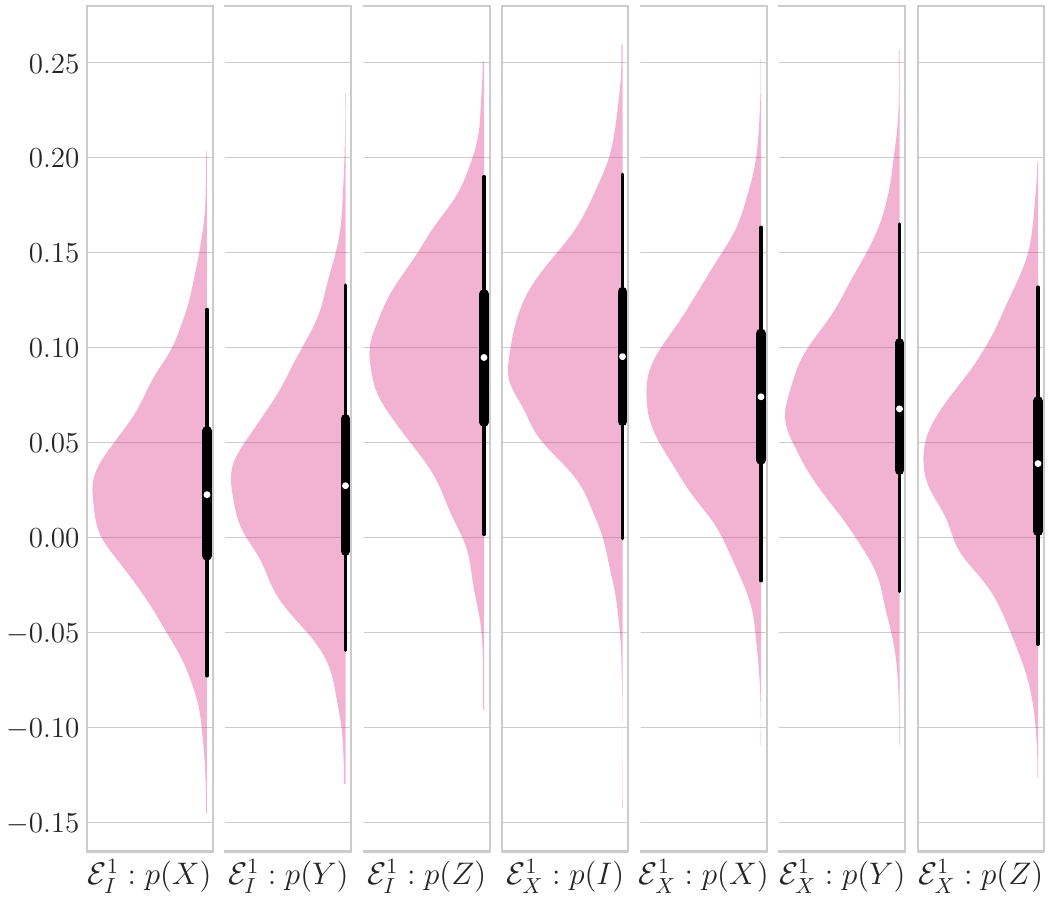}
		\caption{ \textbf{ $\E^{1}$ } with $p_{\text{H1-1}}(\d=1) = 0.0104 \pm 0.0002$.  }
    \end{subfigure}

	\vspace{5mm}

\begin{subfigure}{0.99\textwidth}
		\includegraphics[height=6cm]{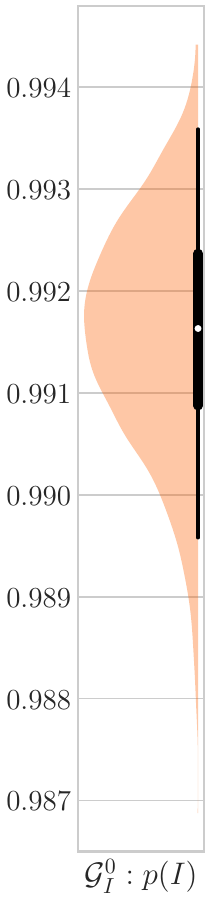}
		\includegraphics[height=6cm]{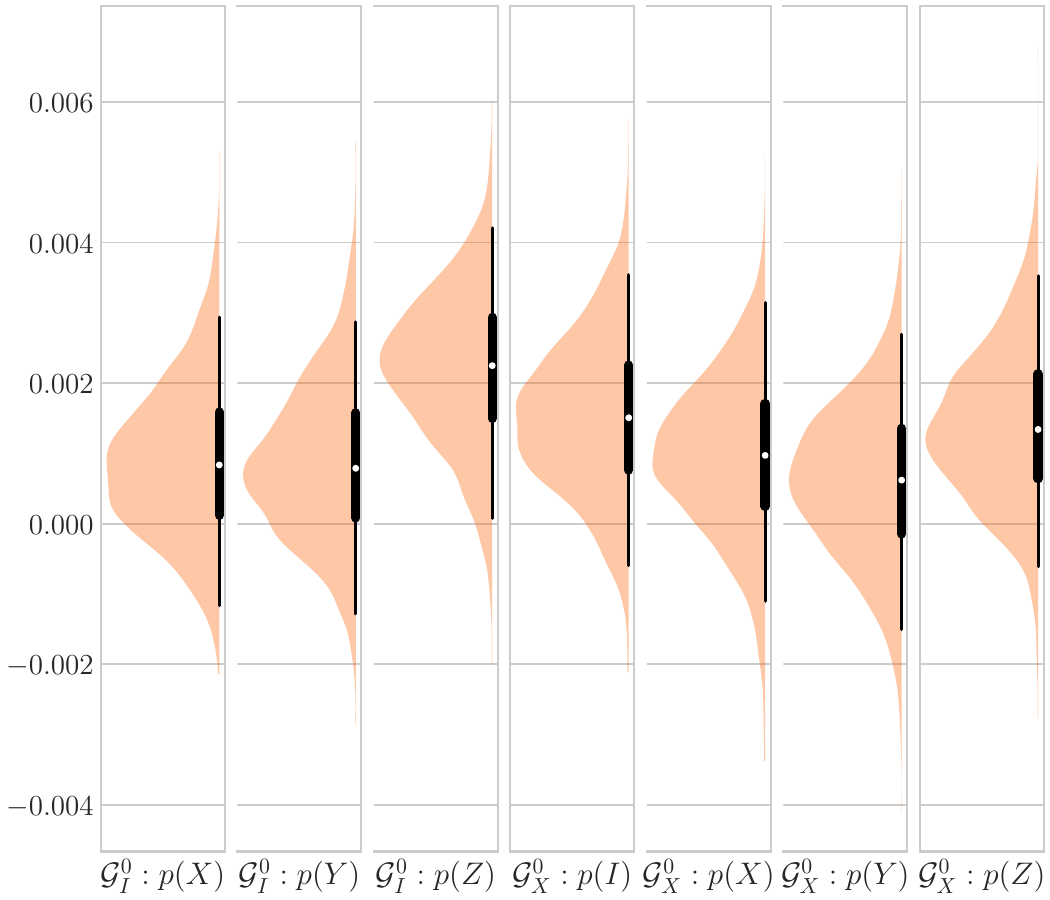}
		\caption{ \textbf{$\G(\rho) := \sum_{\m} \P^{\m}(\rho) = \G_{I}(\rho) + IX \G_{X}(\rho) IX $}, average channel per gadget.  }
    \end{subfigure}

   \caption{\textbf{Bayesian LSD-DRT for the $\nkd{2}{1}{1}$ code.}
   Marginal posterior distributions for the physical channels $\E^{\d}(\rho) := \E^\d_{I}(\rho) + IX \E^\d_{X}(\rho)IX$ are estimated through LSD-DRT (see Figure~\ref{fig:bayes-model}) using uniform priors on all eigenvalues.
   Data gathered using Quantinuum H1-1 with the LSD-DRT protocol, see Table \ref{tab:exp-settings}.  }
   \label{fig:LSD-2,1,1-bayes-uniform-prior}
\end{figure*}

\begin{figure*}[h]
    \begin{subfigure}{0.8\textwidth}
        \centering
        \includegraphics[width=\textwidth]{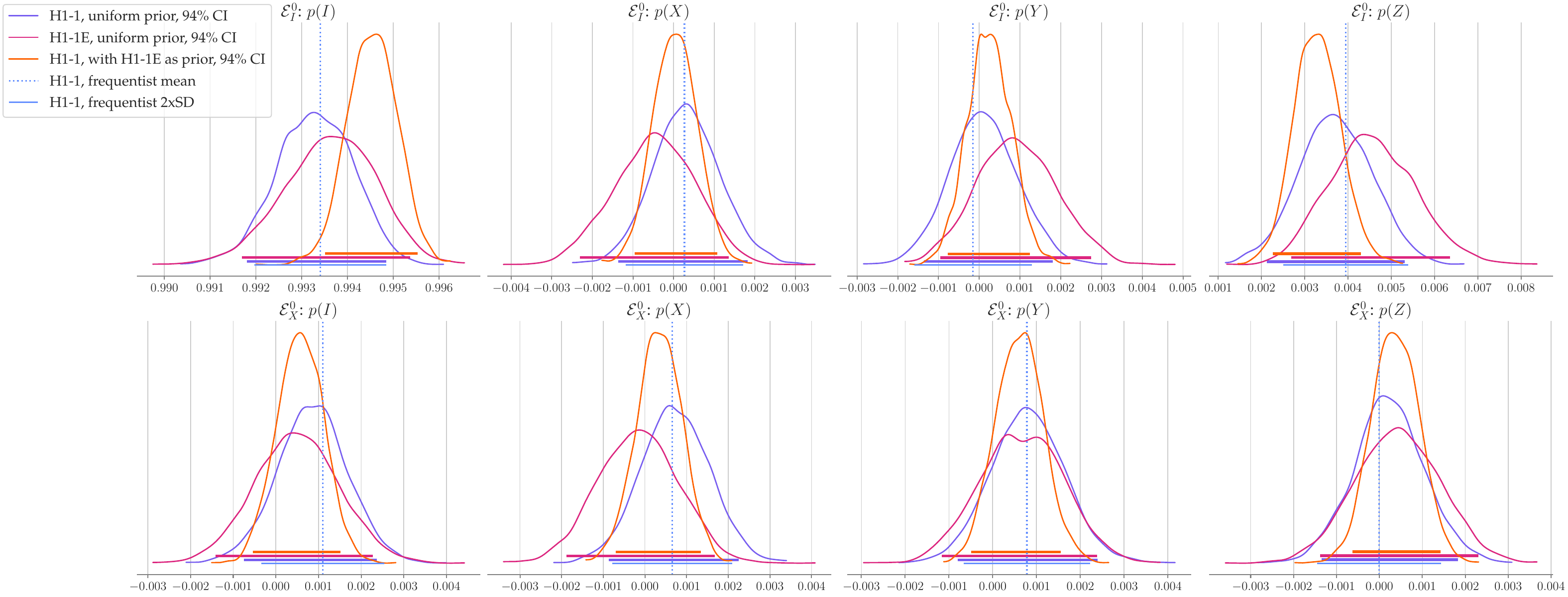}
		\caption{ \textbf{ $\E^{0}$} with $p_{\text{H1-1}}(\d=0) = 0.9896 \pm 0.0002$ and $p_{\text{H1-1E}}(\d=0) = 0.9911 \pm 0.0002$. }
    \end{subfigure}

	\vspace{5mm}

\begin{subfigure}{0.8\textwidth}
		\includegraphics[width=\textwidth]{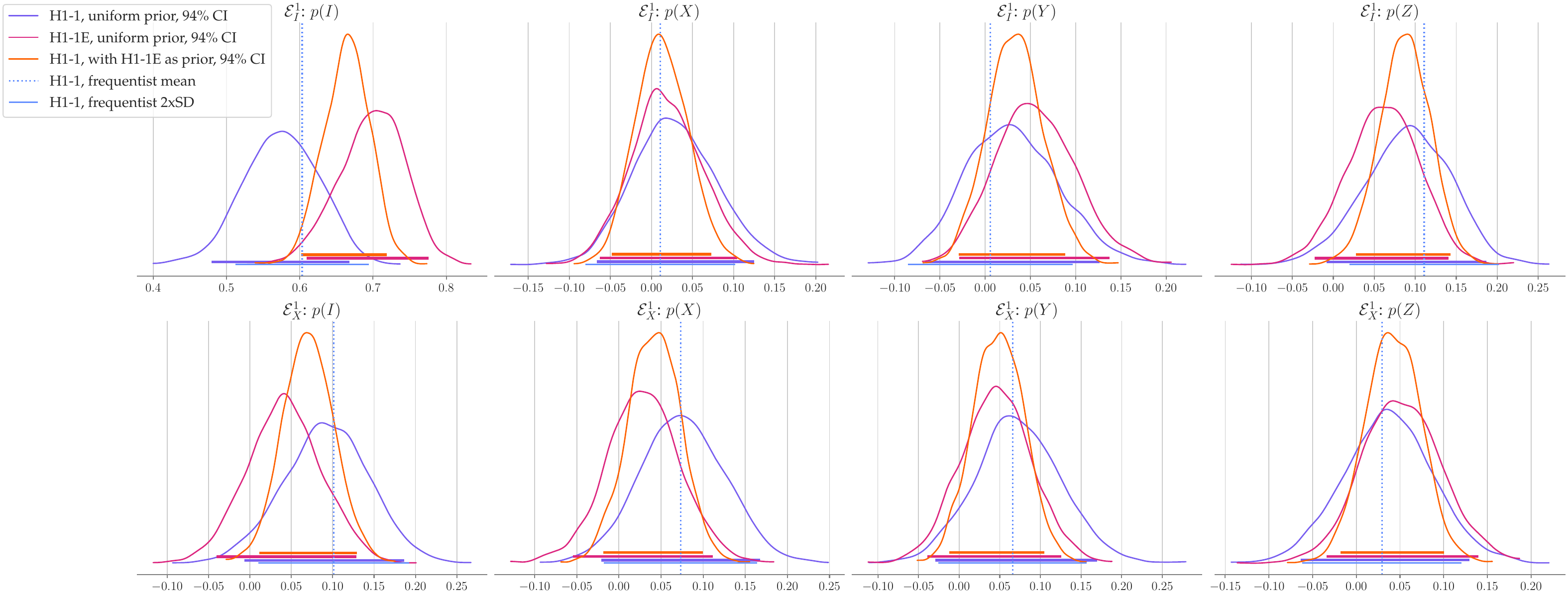}
		\caption{ \textbf{ $\E^{1}$ } with $p_{\text{H1-1}}(\d=1) = 0.0104 \pm 0.0002$ and $p_{\text{H1-1E}}(\d=1) = 0.0089 \pm 0.0002$.  }
    \end{subfigure}

	\vspace{5mm}

\begin{subfigure}{0.8\textwidth}
		\includegraphics[width=\textwidth]{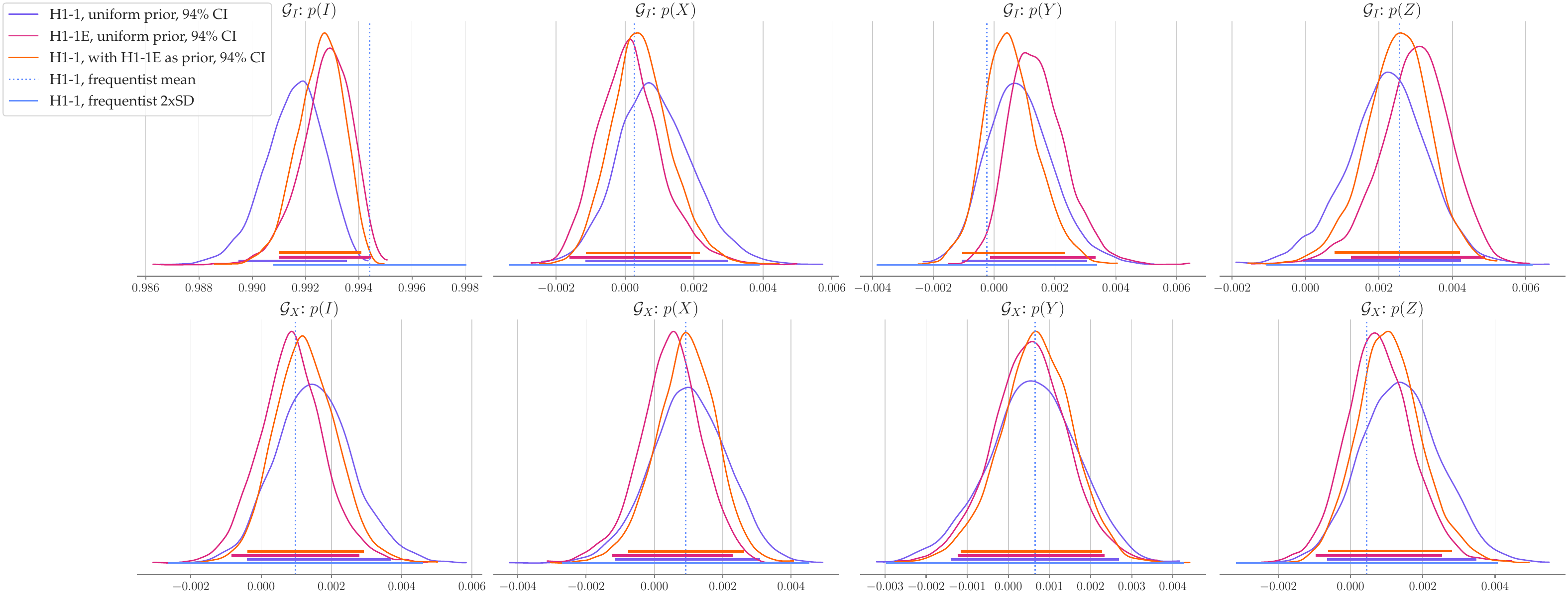}
		\caption{ \textbf{$\G(\rho) := \sum_{\m} \P^{\m}(\rho) = \G_{I}(\rho) + IX \G_{X}(\rho) IX $}, average channel per gadget.  }
    \end{subfigure}

   \caption{\textbf{Bayesian estimation with informative priors applied to the $\nkd{2}{1}{1}$ code.}
   Marginal posterior distributions for the physical channels $\E^{\d}(\rho) := \E^\d_{I}(\rho) + IX \E^\d_{X}(\rho)IX$ are given for two datasets, H1-1 and H1-1E,  with Bayesian post-processing (see Figure~\ref{fig:bayes-model}) using uniform priors on all eigenvalues.
   We then perform alternative post-processing for H1-1 with informative priors; where we use the posterior distributions for all eigenvalues for H1-1E as prior for H1-1.
   Compared to uniform priors, this gives increased precision.
   Data gathered with the LSD-DRT protocol, see Table \ref{tab:exp-settings}.  }
   \label{fig:LSD-2,1,1-bayes}
\end{figure*}

\begin{figure*}[h]
    \begin{subfigure}{0.48\textwidth}
        \centering
        \includegraphics[width=\textwidth]{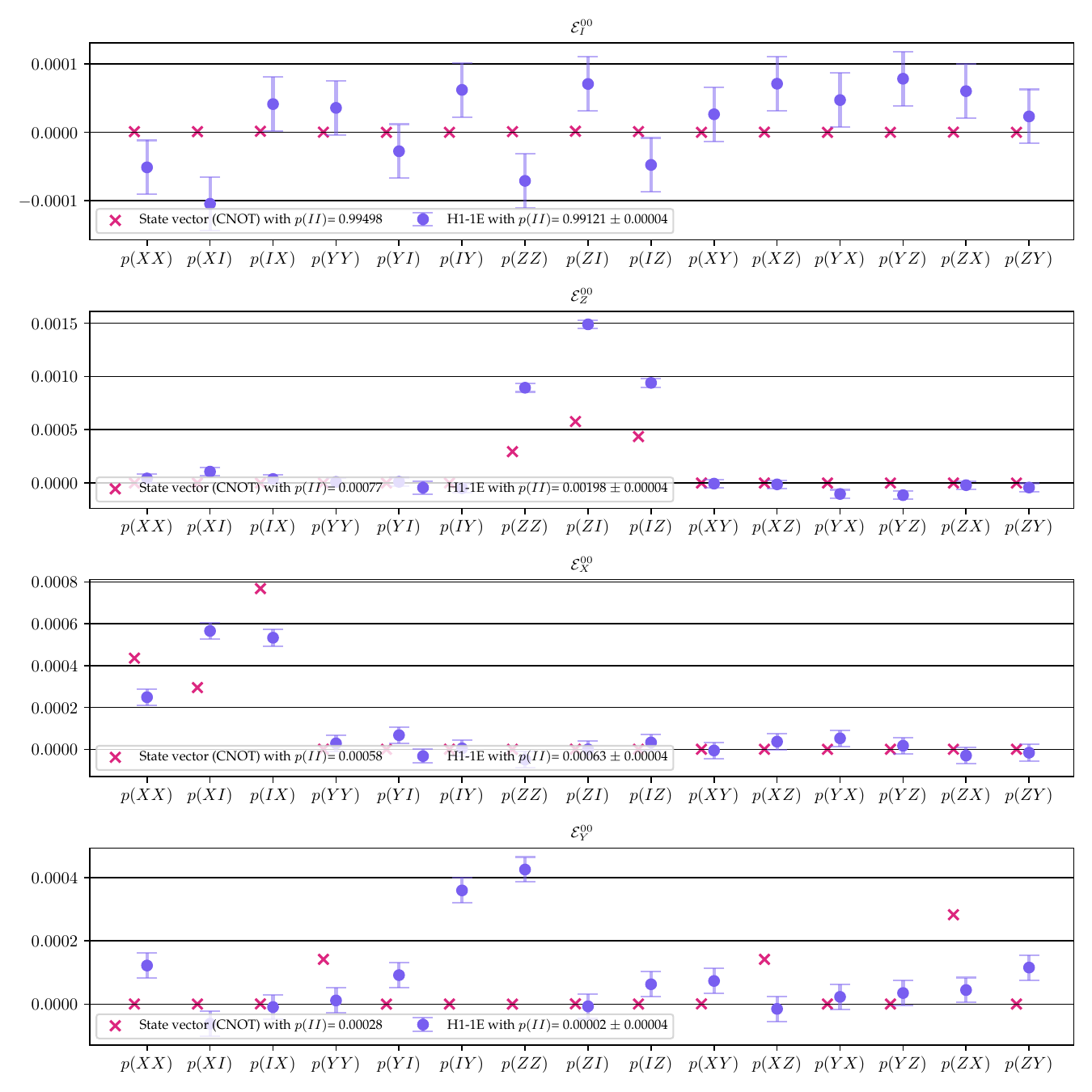}
		\caption{ \textbf{$\E^{00}$} with $p_{\text{H1-1E}}(\d=00) = 0.9672 \pm 0.0001$ and $p_{\text{state vector}}(\d=00) = 0.9882$. }
    \end{subfigure}
	\hspace{5mm}
    \begin{subfigure}{0.48\textwidth}
		\includegraphics[width=\textwidth]{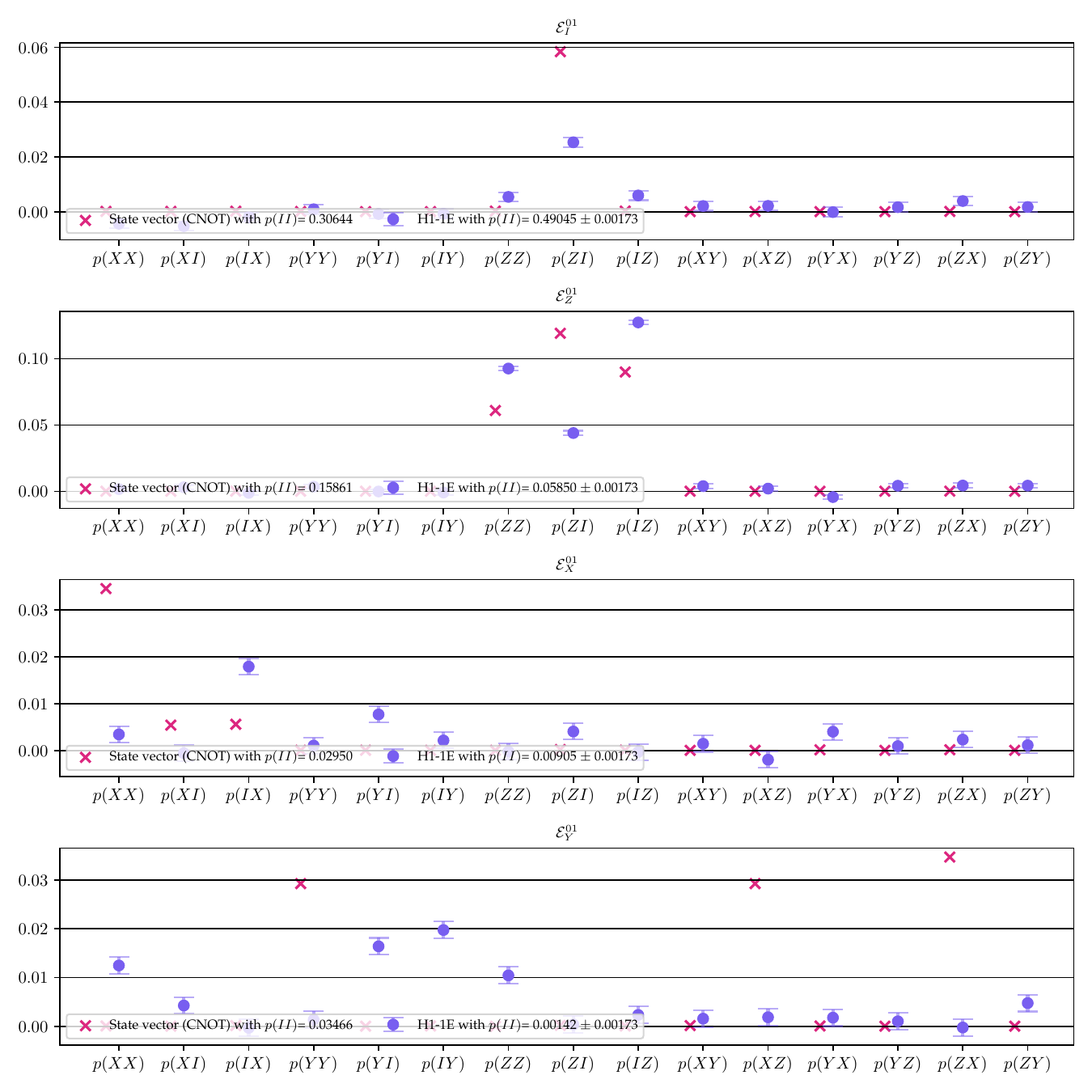}
		\caption{ \textbf{ $\E^{01}$} with $p_{\text{H1-1E}}(\d=01) = 0.0160 \pm 0.0001$ and $p_{\text{state vector}}(\d=01) = 0.0048$.   }
    \end{subfigure}

	\vspace{5mm}

	\begin{subfigure}{0.48\textwidth}
        \centering
        \includegraphics[width=\textwidth]{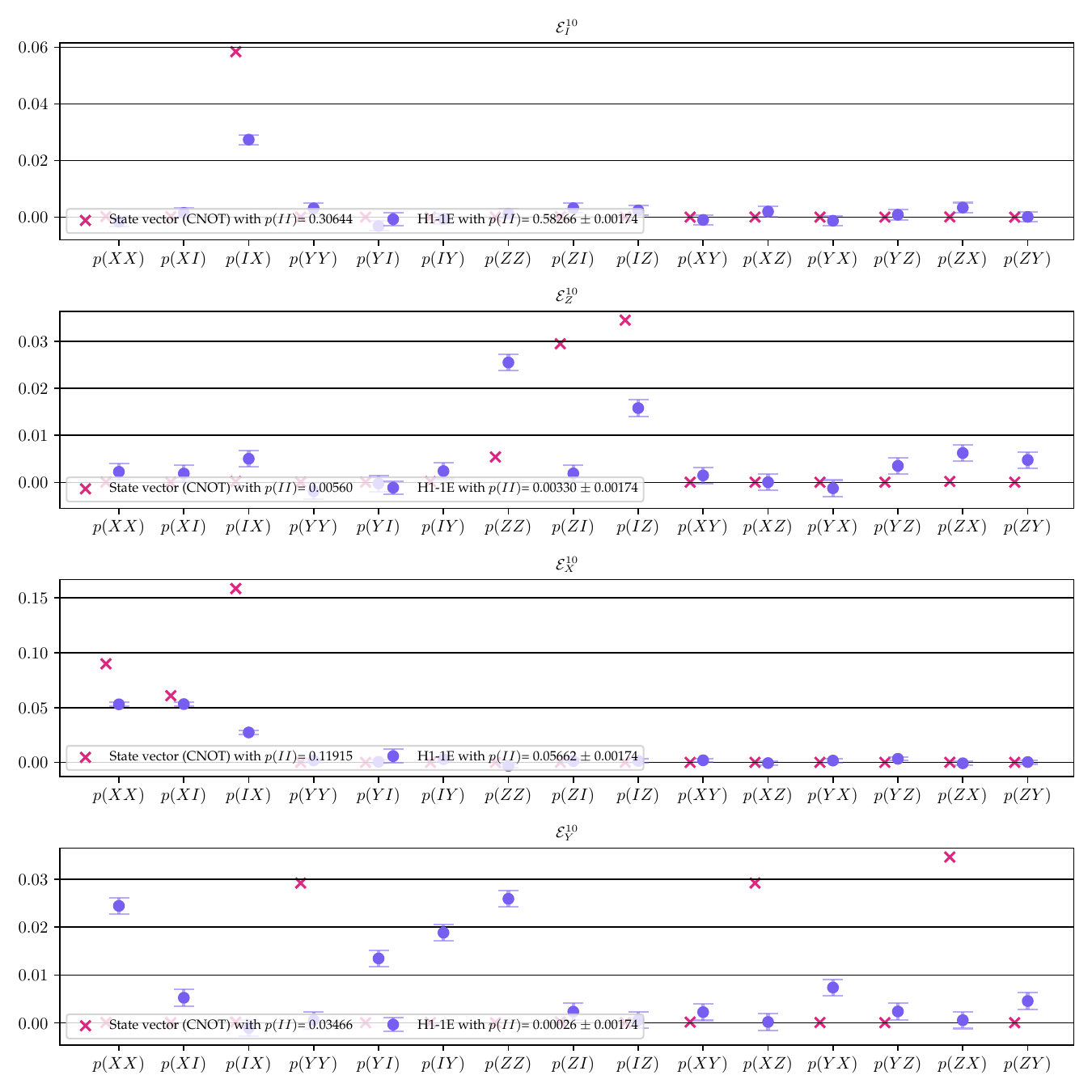}
		\caption{ \textbf{ $\E^{10}$} with $p_{\text{H1-1E}}(\d=10) = 0.0121 \pm 0.0001$ and $p_{\text{state vector}}(\d=10) = 0.0048$.   }
    \end{subfigure}
	\hspace{5mm}
    \begin{subfigure}{0.48\textwidth}
		\includegraphics[width=\textwidth]{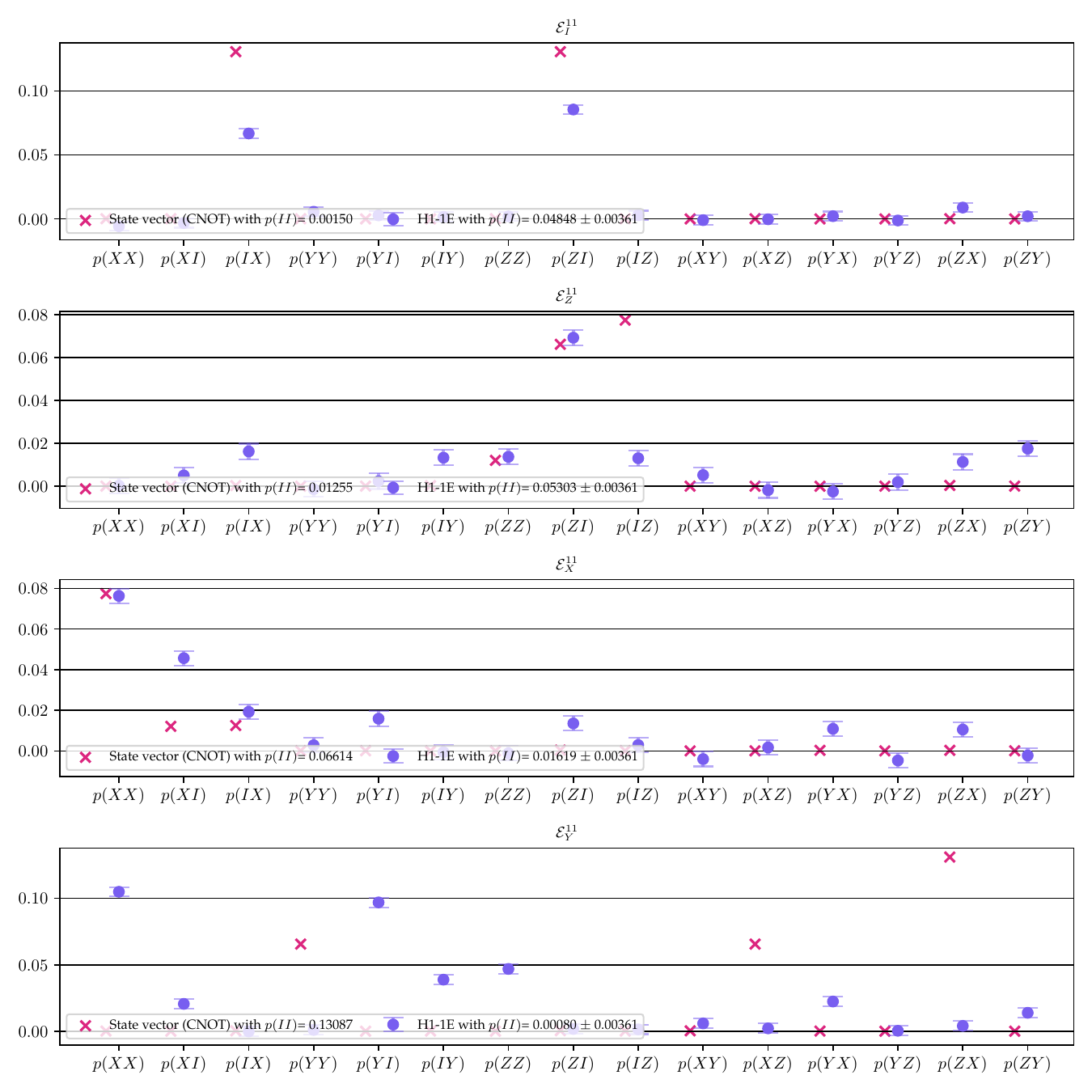}
		\caption{ \textbf{ $\E^{11}$} with $p_{\text{H1-1E}}(\d=11) = 0.0047 \pm 0.0001$ and $p_{\text{state vector}}(\d=10) = 0.0021$.   }
    \end{subfigure}

   \caption{\textbf{Frequentist LSD-DRT for the $\nkd{4}{2}{2}$ code.}
   The four physical channels $\{\E^{\d}\}$ are estimated, each with four logical components $\E^{\d}(\rho) := \E^\d_{I}(\rho) + IIIX \E^\d_{X}(\rho)IIIX + IIIZ \E^\d_{Z}(\rho)IIIZ + IIIY \E^\d_{Y}(\rho)IIIY$.
   Data for H1-1E is gathered with the LSD-DRT protocol, see Table \ref{tab:exp-settings}, whereas numerical simulations with a CNOT error model are conducted without finite sampling.
   }
   \label{fig:LSD-4,2,2-frequentist}
\end{figure*}

\begin{figure*}[h]
    \includegraphics[width=0.48\textwidth]{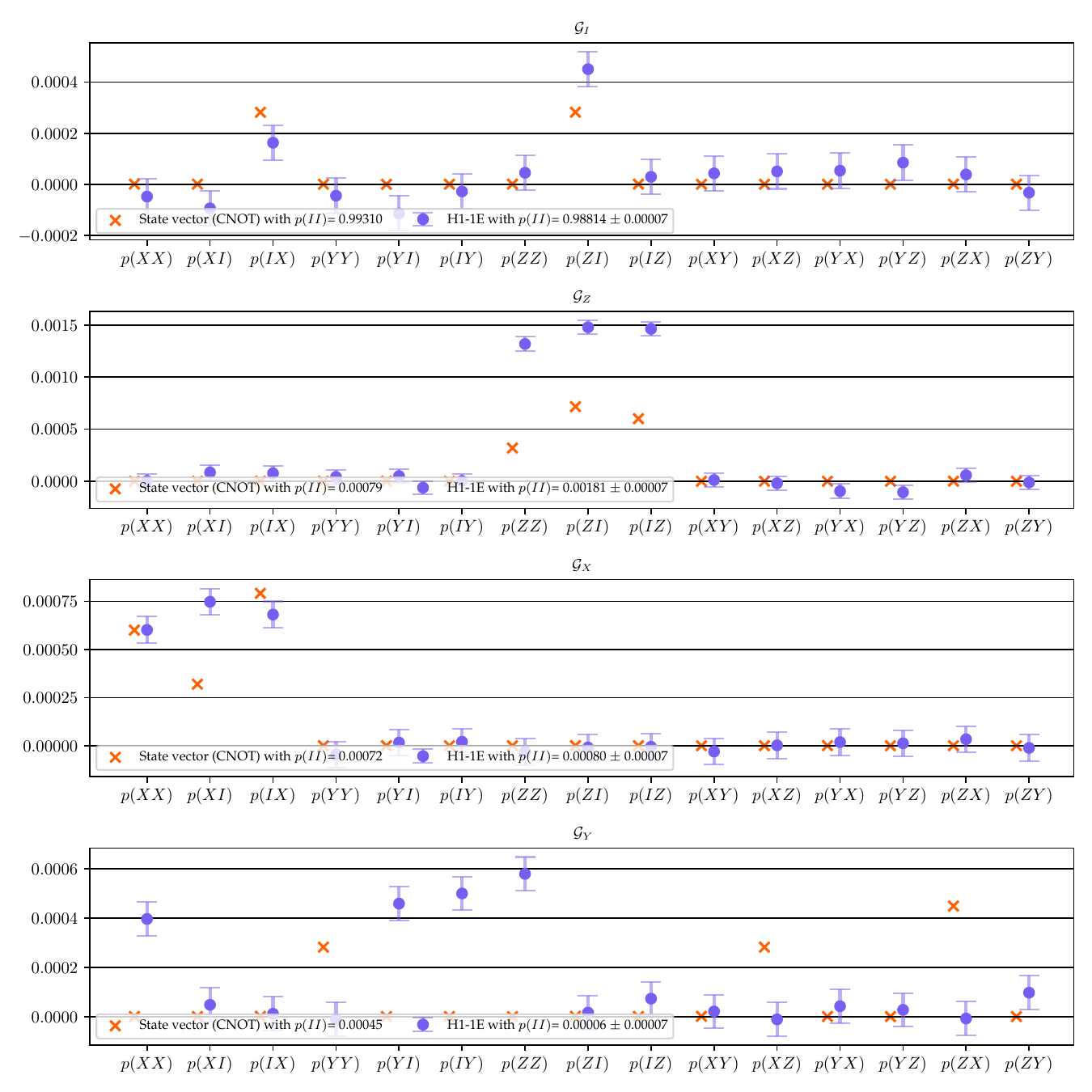}
   \caption{\textbf{Frequentist LSD-DRT for the $\nkd{4}{2}{2}$ code, average case.}
   \textbf{$\G(\rho) := \sum_{\m} \P^{\m}(\rho) = \G_{I}(\rho) + IX \G_{X}(\rho) IX $}, the average channel per gadget is estimated.
   Data for H1-1E is gathered with the LSD-DRT protocol, see Table \ref{tab:exp-settings}, whereas numerical simulations with a CNOT error model are conducted without finite sampling.
   }
   \label{fig:LSD-4,2,2-frequentist-average}
\end{figure*}

\begin{figure}[h]
    \begin{subfigure}{0.47\textwidth}
        \centering
        \begin{quantikz}[row sep=0.2cm, column sep=0.1cm]
            \lstick{$0$}			& \targ{} 	& \qw 		& \ctrl{4}	& \qw 		& \qw 		& \qw 		& \qw 		& \qw  		& \rstick{$0$} \qw  \\
            \lstick{$1$}			& \qw 		& \ctrl{3} 	& \qw		& \targ{} 	& \qw 		& \qw 		& \qw 		& \qw 		& \rstick{$1$} \qw  \\
            \lstick{$2$}			& \qw 		& \qw 		& \qw 		& \qw	 	& \qw 		& \targ{} 	& \ctrl{2}  &  \qw		& \rstick{$2$} \qw  \\
            \lstick{$3$}			& \qw 		& \qw 		& \qw 		& \qw 		& \ctrl{1} 	& \qw 		& \qw 		& \targ{}   & \rstick{$3$} \qw  \\
            \lstick{$\ket{0}$}  & \qw 		&  \targ{}  & \targ{} 	& \qw 		& \targ{} 	& \qw 		& \targ{} 	& \qw	  	&  \meter{} \rstick{$\scriptstyle \pm Z$}  \\
            \lstick{$\ket{+}$}  & \ctrl{-5} & \qw 		& \qw		& \ctrl{-4} & \qw 		& \ctrl{-3} & \qw 	 	& \ctrl{-2} &  \meter{} \rstick{$\scriptstyle \pm X$}  
		\end{quantikz}
            \caption{\textbf{Gadget of depth 4.}
            A weight-1 error on the auxiliary qubits before the first 4 gates is flagged by the syndrome measurements.
            A weight-1 error before the final auxiliary gates propagates to a weight-1 error on the data qubits, which is caught by the next round of error detection.}
            \label{fig:fault-tolerant-422-gadget}
    \end{subfigure}
	\hspace{0.2cm}
    \begin{subfigure}{0.47\textwidth}
        \begin{quantikz}[row sep=0.2cm, column sep=0.1cm]
            \lstick{0}			& \targ{} 	& \qw 		& \ctrl{4}	& \qw 		& \qw 		& \qw 		& \qw 		& \qw 		& \qw  		& \qw 		& \rstick{0} \qw  \\
            \lstick{1}			& \qw     	& \ctrl{3} 	& \qw		& \targ{} 	& \qw 		& \qw 		& \qw 		& \qw 		& \qw 		& \qw 		& \rstick{1} \qw  \\
            \lstick{2}			& \qw     	& \qw 		& \qw 		& \qw	 	& \qw  	    & \targ{} 	& \qw		& \targ{} 	& \qw		& \ctrl{2} 	& \meter{} \rstick{$\scriptstyle \pm Z$}  \\
            \lstick{3}			& \qw 	  	& \qw 		& \qw 		& \qw 		& \ctrl{1} 	& \qw 		& \ctrl{2} 	& \qw 		& \targ{}   & \qw 		& \meter{} \rstick{$\scriptstyle \pm X$}  \\
            \lstick{$\ket{0}$}  & \qw 		&  \targ{}  & \targ{} 	& \qw 		& \targ{} 	& \qw 		& \qw 		& \ctrl{-2} & \qw	  	& \targ{} 	& \rstick{2} \qw  \\
            \lstick{$\ket{+}$}  & \ctrl{-5} & \qw 		& \qw		& \ctrl{-4} & \qw 		& \ctrl{-3} & \targ{} 	& \qw 		& \ctrl{-2} & \qw 		& \rstick{3} \qw
            \end{quantikz}
            \caption{\textbf{Leakage-protected (LP) gadget of depth 5.}
            Similar to Figure~\ref{fig:lp-fault-tolerant-211-gadget}, we swap out physical qubits.}
            \label{fig:lp-fault-tolerant-422-gadget}
    \end{subfigure}
   \caption{\textbf{Fault tolerant flag syndrome extraction gadgets for the $\nkd{4}{2}{2}$ code.}}\label{fig:422-gadgets}
\end{figure}
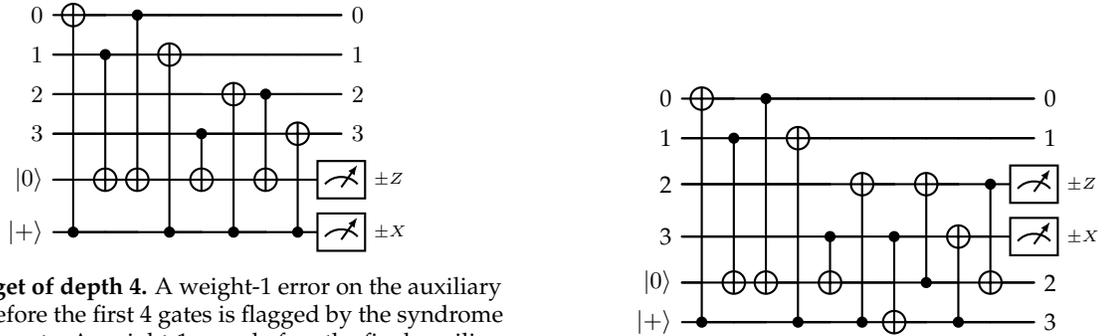

\end{document}